\documentclass[12pt,twoside]{article}
\title{A communication-efficient, online changepoint detection method for monitoring distributed sensor networks}

\usepackage[left = 2.5cm, right = 2.5cm, top = 2cm, bottom = 2.5cm]{geometry} 
\usepackage{setspace}
\usepackage{tikz}
\usetikzlibrary{positioning, shapes}
\usepackage{authblk}
\usepackage{graphicx}
\usepackage{amsfonts}
\usepackage{xcolor}
\usepackage{amssymb}

\newcommand{\blue}{\color{blue}}
\usepackage{amsmath}
\usepackage{float}
\usepackage{csquotes}
\restylefloat{table}
\usepackage{adjustbox}
\usepackage{lscape}
\usepackage{natbib} 
\usepackage{placeins}
\usepackage{caption}
\usepackage{amsthm}
\numberwithin{equation}{section}
\usepackage[ruled,vlined]{algorithm2e}
\usepackage{titlepic}
\DeclareMathAlphabet{\mathpzc}{OT1}{pzc}{m}{it}
\usepackage{color}
\usepackage[toc,page]{appendix}
\usepackage{fancyhdr} 
\usepackage{url}
\usepackage{subfig}
\usepackage{xcolor}
\usepackage{bbm}
\usepackage{verbatim}
\usepackage{multirow}
\usepackage{scrextend}
\usepackage{enumitem}   

\usepackage{csquotes}
\renewcommand{\mkbegdispquote}[2]{\itshape}
\newtheorem{Thm}{Theorem}
\newtheorem{Assum}{Assumption}

\newtheorem{corollary}{Corollary}[Thm]
\usepackage[utf8]{inputenc} 
\usepackage[T1]{fontenc}

\usepackage{multicol}

\usepackage{tikz}
\newcommand*\ab{.4}
\tikzset{
  net node/.style = {circle, minimum width=2*\ab cm, inner sep=0pt, outer sep=0pt, draw=black, fill=white},
  net sqnode/.style = {square, minimum width=2*\ab cm, inner sep=0pt, outer sep=0pt, draw=black, fill=white},
  net root/.style = {cloud, cloud puffs=15.7, cloud ignores aspect, minimum width=10*\ab cm, minimum height=4*\ab cm, align= center, draw=black, line width =1pt},
  net connect/.style = {line width=1pt, draw=black},
}

\usepackage[pdftex]{hyperref}  
\usepackage[]{algorithm2e}
\author[1]{Ziyang Yang}
\author[2]{Idris A. Eckley}
\author[3]{Paul Fearnhead}
\affil[1]{STOR-i Doctoral Training Centre, Lancaster University}
\affil[2,3]{Department of Mathematics and Statistics, Lancaster University}

\graphicspath{{Figures/}}

\begin{document}
\bibliographystyle{apalike}  
\maketitle
\begin{abstract}
We consider the challenge of efficiently detecting changes within a network of sensors, where we also need to minimise communication between sensors and the cloud. We propose an online, communication-efficient method to detect such changes. The procedure works by performing likelihood ratio tests at each time point, and two thresholds are chosen to filter unimportant test statistics and make decisions based on the aggregated test statistics respectively. We provide asymptotic theory concerning consistency and the asymptotic distribution if there are no changes. Simulation results suggest that our method can achieve similar performance to the idealised setting, where we have no constraints on communication between sensors, but substantially reduce the transmission costs.

\end{abstract}
\section{Introduction}
\label{sec:intro}

During the last decade, there has been a significant focus on the important challenge of efficient and accurate detection of changes in both univariate and multivariate data sequences \cite[][]{chofryz2015, fisch2018linear, kovacs2023seeded,truong2020selective, tveten2022scalable, wang2018high}. More recently, focus has turned to translating the efficiency of such approaches to the online setting, typically motivated by an applied challenge such as how to deal with limited computational power \cite[e.g.][]{ward2023}. Recent major contributions to the online setting include \cite{adams2007bayesian}, \cite{tartakovsky2014sequential}, \cite{yu2020note}, \cite{chen2020highdimensional}, and \cite{romano2023fast}. In this paper we consider a less studied scenario, monitoring edge-behaviour within  distributed sensor networks, which are common architectures within the Internet of Things framework (IoT). The importance of efficiently detecting changes at the edge efficiently, whilst minimising communication between sensors and the cloud is perhaps best appreciated by considering two key applications: detecting cyber-attacks on smart cities \citep{8666450} and optimising the performance of base stations \citep{Wu2019EdgeCI}.\\

Consider, by way of example, Figure \ref{fig:schematic} which shows a schematic representation of real-time monitoring within a distributed network. Here we assume that $d$ data streams are monitored, each by its own sensor. Communication between the sensors and the centre is possible as shown by the dashed lines. An unusual event happens at time $\tau$, and we want to detect this event as quickly as possible. However, in modern sensor networks that deploy IoT devices the computational resources of the sensors can be substantial. Moreover, communication between the sensors and the cloud can be problematic due to the heavy energy usage involved with transmitting data \citep{7796149,10.1145/3154384}. As such, we need algorithms that can identify the time when it is important for information to be shared with the cloud.
%
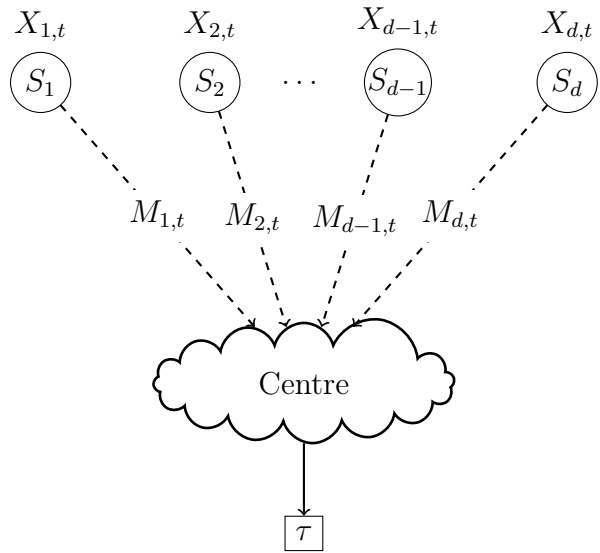
\begin{figure}[h]
  \centering
  \begin{tikzpicture}
    \node (root) [net root] (centre)  at (0,0) {Centre};
    \node[label=above:{$X_{1,t}$}] [net node] (node) (s_1) at (-3.5,4) {$S_1$};
    \node[label=above:{$X_{2,t}$}] (node) [net node] (s_2) at (-1.25,4) {$S_2$};
    \node[label=above:{$X_{d-1,t}$}] (node) [net node] (s_3) at (1.25,4) {$S_{d-1}$};
    \node[label=above:{$X_{d,t}$}] (node) [net node] (s_4) at (3.5,4) {$S_{d}$};
    \path (s_2) -- node[auto=false]{\ldots} (s_3);
    \draw[->, thick, dashed] (s_1) -- ++(centre) node [midway,fill=white] {$M_{1,t}$};
    \draw[->, thick, dashed] (s_2) -- ++(centre) node [midway,fill=white] {$M_{2,t}$};
    \draw[->, thick, dashed] (s_3) -- ++(centre) node [midway,fill=white] {$M_{d-1,t}$};
    \draw[->, thick, dashed] (s_4) -- ++(centre) node [midway,fill=white] {$M_{d,t}$};
    \node (tau) at (0,-2) [draw, rectangle] {$\tau$};
    \draw[->, thick] (centre) -- ++(tau);        
  \end{tikzpicture}
  \caption{Schematic representation of a sensor network made up of $d$ sensors, where $S_i$ is the index for sensor $i$, $X_{i,t}$ is the data observed at sensor $i$, and $M_{i,t}$ is the message transmitted from sensor $i$ to centre at time $t$.}
  \label{fig:schematic}    
\end{figure}
\noindent More specifically, in this article, we seek to develop a new method to detect changes within such a network in real time with high statistical power and as little communication and computation as possible. \\

Changepoint methods which can be applied in the fully centralised problem, when the data from the sensors is processed and transmitted to the centre (cloud) at every time step, are well studied. Approaches typically seek to calculate the maximum or the sum of all the test statistics \cite[see, e.g., ][]{10.1093/biomet/asq010, xie2013, chan2017, chen2020highdimensional, gosmann2020sequential}. The rationale behind these methods is to set thresholds and raise the alarm if the aggregated test statistics from multiple streams exceed pre-defined thresholds. Numerical experiments \cite[][]{10.1093/biomet/asq010} indicate that taking the maximum is the optimal method when there are only a few affected data streams - what we will term a sparse change. Conversely, taking the sum is optimal when most data streams are affected, also known as a dense change. \\


Recent contributions to this distributed problem include \citep{489500, 923755, 1413462, 1459066, 6034390, tartakovsky2006performance, 7105922}. Among them, two recent papers of particular interest develop \textit{communication efficient} schemes for monitoring a large number of data streams \citep{doi:10.1080/07474946.2018.1548849,liu2016scalable}. The key idea is that each sensor computes a local monitoring statistic and then employs a thresholding step, only sending the statistic to the centre if there is some evidence of a change. The information from multiple sensors is then combined at the centre. This approach reduces unnecessary transmission by ignoring streams with little evidence for a change, while only focusing on data streams that show signs of change. \\

Although computationally feasible, existing works assume that the pre- and post-change mean are known. In practice, the pre-change mean can be estimated based on historical data. However assuming a known post-change mean is typically unrealistic in practice, with an incorrect value potentially leading to a failure to detect, or poor detection power. Liu et al (2019) approximate the post-change mean recursively but, as a consequence, somewhat sacrifice statistical power of the algorithm. \\


Our approach builds on recent work developing the moving sum (MOSUM) as a window-based changepoint methods \citep[see, e.g.,][]{AUE20122271, kirch2015, kirch2018}.  Specifically, we propose an online communication-efficient changepoint detection algorithm (distributed MOSUM) to detect changes in real-time within the distributed network setting. A local threshold is chosen to filter out unimportant information and only transmit the statistically important test statistic to the centre. The change will be alarmed when the aggregated test statistic exceeds the pre-defined global threshold in the central cloud. 
The low time complexity and communication efficient scheme of our proposed method makes it suitable for online monitoring. We also establish that the proposed method can achieve similar statistical power as the idealistic setting, where there is no communication constraint, at detecting large changes whilst substantially reducing the transmission cost. Moreover, we also show how to make the detection performance of distributed MOSUM close to that of the idealised setting by increasing the window size, which will only sacrifice the storage cost and a little transmission cost.\\

The key differences between our work and previous distributed changepoint detection contributions \citep[e.g.,][]{liu2016scalable} are: Firstly, a moving window-based test statistic MOSUM is chosen to avoid the requirement of knowledge of the post-change mean. Secondly, earlier works have been based on the framework that controls the average run length (ARL) - the average amount of time until incorrectly detect a change. However, such a metric gives a somewhat limited amount of information since the distribution of run length is usually unknown. For instance, if multiple procedures end quickly while a few replications stop significantly longer, the ARL would be the same if all the replications terminated around the same time. Conversely, in this work, we present methods in terms of controlling the error rate under the null at a specific level, and with asymptotic power 1 under alternatives. Furthermore, our ideas generalise trivially to methods controlling the average run length.\\

The structure of this paper is as follows. In Section \ref{sec:problem}, the problem setting is outlined, before introducing the distributed MOSUM methodology in Section \ref{sec:methodology}. Several theoretical results for this new approach are given in Section \ref{sec:theory}. Simulation studies are carried out in Section \ref{sec:sims}, before ending with some concluding remarks (Section \ref{conclusion}).
\section{Problem setting} \label{sec:problem}
We begin by assuming that we have $d$ sensors,
each of which is observed as follows: $\mathbf{X_t}=(X_{1,t},X_{2,t},X_{3,t}...,X_{d,t})$ at every time point $t \in \mathbb{N}$. Here $\mathbf{X_t}$ could be raw data or the residuals after pre-processing the data. These observations are assumed to be identically distributed and independent across series. Such assumptions are common in the problem of detecting changes within a distributed system setting \citep{1021129,  10.1093/biomet/asq010,  xie2013, liu2016scalable}. We do not strictly assume time independence here, but our method is optimal when this assumption holds. Moreover, the impact of time dependence will be numerically studied in Section \ref{sec:sim-ind}. \\

We begin by assuming that at some unknown time, $\tau$, the distribution of some unknown subsets of $d$ sensors will change.  For simplicity, we only consider change in mean, but the ideas below are easily extended to other changepoint settings. Therefore, in this illustrative change in mean setting, the model for the data is expressed as follows:
\begin{align}
  \label{eq:model}
  X_{i,t} = \mu_i + \delta_{i} \mathbbm{1}_{ \{ t > \tau \}} + \epsilon_{i,t}, \quad t \in \mathbbm{N}, 1 \leq i \leq d,
\end{align}
where $\mu_i$ is the known pre-change mean, $\delta_i$ is the mean shift, and $\{\epsilon_{i,t}: t \in \mathbbm{N}\}$ are strictly stationary error sequences. After time $\tau$, the mean of the $i$-th data stream changes immediately from $\mu_i$ to $\mu_i + \delta_i$. Here it is useful to note that our setting also permits some $\delta_i=0$, which means that only a subset of data streams are affected by the change. Without loss of generality, we assume $\mu_i=0$. Under the null hypothesis, the model for the data can be rewritten as
\begin{align}
  \label{eq:model_eq}
  X_{i,t} = \epsilon_{i,t}, \quad t \in \mathbbm{N}, 1 \leq i \leq d.
\end{align}
Moreover under the alternative hypothesis, the model is $X_{i,t} = \delta_i+\epsilon_{i,t}, t \in \mathbbm{N}, 1 \leq i \leq d$. Our aim is to monitor such a system and raise the alarm as soon as possible following the event at time $\tau$. One way of achieving this is to perform hypothesis testing sequentially, i.e., evaluate the null hypothesis of no change in mean at each time point $t \in \mathbb{N}$. The algorithm will stop and declare a change when we can reject the null hypothesis.\\

In the classical sequential changepoint detection problem, we evaluate the performance of an algorithm subject to a constraint on its false alarm rate. First, consider an open-ended stopping rule where the algorithm never we have an infinite time-window of measurements and the algorithm never halts until it detects a change. The false alarm rate can be evaluated in two ways. Assume there is no change, and let $\widehat{\tau}$ be the time at which we detect a change, with the convention that $\widehat{\tau}=\infty$ if we detect no change. One approach is to control the average run length, $E^{\infty}(\widehat{\tau})$, the expected time of to a false alarm. This makes sense for procedures with a constant threshold for detection, for which we are certain to detect a change under the Null if we monitor for an infinite time period. Alternatively, one can control the false alarm rate, $P^\infty(\widehat{\tau} <\infty)$, the probability of a false alarm. To control this over an infinite time horizon requires increasing the threshold for detecting a change over time. Equivalently, this can be achieved by multiplying the test statistic with a weight function  $w(\cdot)<1$. See \cite{10.2307/3533257, doi:10.1002/jae.776, doi:10.1080/07474940801989087, AUE20122271, kirch2015, Weber2017_1000068981, Yau_sinica, kirch2018, kengne2020inference} for examples of how to choose an appropriate weight function. \\

In our paper, we focus on controlling the false alarm rate. However \cite{AUE20122271} states that "applying open-ended procedures built from the asymptotic critical values have a tendency to be too conservative in ﬁnite samples". Therefore, our paper considers a close-ended stopping rule. In this approach, the algorithm will stop either upon detecting a change or upon reaching the predefined monitoring time $T$. We thus control the false alarm rate over a time time window of length $T$. However, the ideas we present can easily be adapted to the open-ended setting, and also to methods which control the average run length.\\

Under the context of distributed changepoint detection problem, we additionally evaluate the index - the average transmission cost $\bar \Delta$. This is the average number of transmissions at each time step for $d$ sensors, and should be smaller than the pre-specified transmission cost $\Delta$. \\

Before introducing our proposed method, we first review relevant work. At time $t$, the local monitoring statistic, $\mathcal{T}_{i}$ is calculated for the $i$th stream. Then all the local statistics $\mathcal{T}_{i}$ can be combined into a global monitoring statistic $\mathcal{T}$ at the fusion centre. There are two common choices of message combinations for monitoring changes within the distributed system. One of these two types, the SUM scheme \citep{10.1093/biomet/asq010}, declares a change when the sum of all the local monitoring statistics exceeds a pre-defined threshold, that is: 
\begin{align*}
    \hat \tau_{\mathrm{sum}}(c_{\mathrm{Global}})=\inf\left\{t\geq 1: \mathcal{T} \geq c_{\mathrm{Global}} \right\}=\inf\left\{t\geq 1: \sum_{i=1}^d \mathcal{T}_{i} \geq c_{\mathrm{Global}} \right\},
\end{align*}
where $c_{\mathrm{Global}}$ is global threshold. This way of combining statistics across streams is known to be good if the series are independent and the changes are dense. However, implementing this method on the distributed system requires sending every $\mathcal{T}_{i}$ to the fusion centre, which is expensive. A sum-shrinkage method \citep{liu2016scalable} is proposed to reduce the communication cost by thresholding the test statistics before summing them:
\begin{align*}
   \hat \tau_{\mathrm{sum}}(c_{\mathrm{Local}}, c_{\mathrm{Global}})=\inf\left\{t\geq 1: \mathcal{T} \geq c_{\mathrm{Global}} \right\}=\inf\left\{t\geq 1: \sum_{i=1}^d \mathcal{T}_{i}\mathbbm{I}(\mathcal{T}_{i}\geq c_{\mathrm{Local}}) \geq c_{\mathrm{Global}} \right\}. 
\end{align*}
 Empirically the sum-shrinkage method could achieve similar performance as the SUM scheme in the dense case and surprisingly performs better in the sparse case. \\

When the change is sparse, it has been shown both theoretically and empirically \citep{10.1093/biomet/asq010, liu2016scalable,chen2020highdimensional} that monitoring the maximum of the test-statistics across series is best. In such a setting, the MAX procedure \citep{1021129} monitors the maximum of test statistics and raises the alarm when the maximum of the local test statistics exceeds the thresholds, that is: 
\begin{align*}
    \hat \tau_{\mathrm{max}}(c_{\mathrm{Global}})=\inf\left\{t\geq 1: \mathcal{T}\geq c_{\mathrm{Global}} \right\}=\inf\left\{t\geq 1: \max_{ 1\leq i\leq d} \mathcal{T}_{i} \geq  c_{\mathrm{Global}} \right\}.
\end{align*}
The best choice of different schemes depends on the sparsity of changes which is based on the number of affected data streams $p$. This can be made precise if we consider an asymptotic setting where $p\rightarrow\infty$ 
\citep{enikeeva2019high}, and define a change to be sparse if the number of affected streams is $p=o(\sqrt{d})$, and it to be a dense change otherwise.  
A recent paper \citep{chen2020highdimensional} combines both SUM procedure and MAX procedure to achieve good performance regardless of the sparsity. In the context of distributed monitoring, the MAX procedure is trivially implemented without any communication. Specifically, each sensor has the threshold for the max-statistic and flags a change if their local statistic is above this threshold. Therefore, within this paper, we only focus on developing a communication-efficient version of the SUM scheme. Our aim is a method that performs well for dense changes, but limits the communication cost.  We will use the SUM scheme as the ideal method to compare against since it has no restrictions on communication.

\section{Distributed change point detection method}
\label{sec:methodology}
Our proposed methodology is summarized in Algorithm \ref{hard-thresholding}, and described in detail below. The method essentially comprises of three steps. The first step involves the parallel local monitoring of each data stream by the sensors. As the monitoring unfolds, messages are occasionally sent from the sensors to the centre to indicate the presence of a potential change. Finally, at the centre these messages are aggregated to find changes that occur across a number of data streams.
\begin{algorithm}[]
\footnotesize
\caption{Centralized and distributed MOSUM}\label{hard-thresholding}
\SetKwInOut{Input}{input}\SetKwInOut{Output}{output}
    \SetAlgoLined
    \Input{historic data $x_{i,t}$ for $i=1,2,...,d$, and $1\leq t\leq m$}
     \BlankLine
    {\color{purple}\emph{Estimating the baseline parameters}}\tcp*[f]{can be done offline}\\
    \For{$i=1$ \KwTo $d$}{estimate $\hat{\mu}_i$ and $\hat \sigma_i$}
     \BlankLine
    \KwData{$x_{i,t}$ for $i=1,2,...,d$ at time $t$}
    \BlankLine
    \While{change is detected or reached the maximum monitoring time $T$}{
    {\color{purple}\emph{Local monitoring}} \tcp*[f]{parallel computing}\\
    
    \For{$i=1$ \KwTo $d$}{ 
    $\mathcal{T}_{i}(m,k,h) = \frac{1}{\hat{\sigma}_i} \left|  \sum_{t = m + k - h + 1}^{m+k}\left( X_{i,t} - \hat{\mu}_i \right)\right|.$}

    {\color{purple}\emph{Message passing}}\\
    \If{$w(k,h)\mathcal{T}_{i}(m,k,h)>c_{\textrm{Local}}$}{
           $M_{i,t}=\mathcal{T}_{i}(m,k,h)$ \tcp*[f]{centralized scheme:set $c_{\textrm{Local}}=0$}\\
    \lElse{$M_{i,t}=0$}}

    {\color{purple}\emph{Global monitoring}}\\
    \If{$w(k,h)\sqrt{\sum_{i}^dM_{i,t}}>c_{\textrm{Global}}$} {
        stop the algorithm \\
        \Output{$\hat \tau =t$}}
    {$t \longleftarrow t+1$\\}
    }
\end{algorithm}
\subsection{Local monitoring}
\subsubsection{Estimating the baseline parameters}
Our sequential testing approach requires a historic data set of length $m$ to estimate the baseline parameters. Theoretical results are obtained later in the paper when $m \rightarrow \infty$. The parameters of interest are the mean of each data stream $\mu_i$ and the variance of the errors $\sigma_i^2$. For the $i$th data stream these estimates are,
\begin{align}
  \label{eq:initial_est}
  \begin{split}    
    \hat{\mu}_i &= \frac{1}{m} \sum_{t=1}^{m} X_{i,t}, \\
    \hat{\sigma}_i^2 &= \frac{1}{m}  \sum_{t=1}^{m} \left( X_{i,t} - \hat{\mu}_i \right)^2.    
  \end{split}
\end{align}
If the errors cannot be assumed to be independent we can estimate the long run variance. This requires specifying a kernel function $K(\cdot)$:
\begin{align}
\label{kernel}
  \hat{\sigma}_i^2 &= \frac{1}{m}  \sum_{t=1}^{m} \left( X_{i,t} - \hat{\mu}_i \right)^2 + 2\sum_{j=1}^{m-1} K\left( \frac{j}{l} \right)\hat{\gamma}_{j}^{(i)}, \\
  \textrm{where } \hat{\gamma}_{j}^{(i)} &= \frac{1}{m - j} \sum_{t=1}^{m-j} \left( X_{i,t} - \hat{\mu}_i \right) \left( X_{i,t+j} - \hat{\mu}_i \right).
\end{align}
In this setting, the Kernel function can be seen as a weighting function for sample covariance $\hat{\gamma}_{j}^{(i)}$. The kernel function must be symmetric and such that $K(0)=1$. Various kernel functions are proposed. Standard kernel functions include Truncated \citep{white1984nonlinear}, Bartlett \citep{newey1986simple} and Parzen \citep{gallant2009nonlinear} amongst others. Among them, the Bartlett kernel is frequently used in Econometrics. This kernel takes the form:
\begin{align*} K_{\textrm{Bartlett}}\left(\frac{j}{l}\right) =
    \begin{cases}
      1-\frac{j}{l}, & \text{for $0 \leq j \leq l-1$},\\
      0, & \text{otherwise.}
    \end{cases}  
\end{align*}
For more details, see \cite{doi:10.1111/j.1467-9892.2012.00796.x}; \cite{kiefer2002heteroskedasticity}; \cite{kiefer2002heteroskedasticityb}. 
\subsubsection{Starting local monitoring}
Once the baseline parameters have been estimated, beginning at time $m + 1$ data $X_{i,m+1}, X_{i,m+2}, \hdots$ are observed sequentially and monitored for a change. This is achieved using a MOSUM statistic which at monitoring time, $k$, takes a window containing the most recent $h$ observations:
\begin{align}
  \label{eq:local_MOSUM_method}
  \mathcal{T}_i(m,k,h) = \frac{1}{\hat{\sigma}_i} \left|  \sum_{t = m + k - h + 1}^{m+k}\left( X_{i,t} - \hat{\mu}_i \right)\right|.
\end{align}
Following \cite{AUE20122271}, the MOSUM statistic will declare a change at time $k$ when the weighted local MOSUM statistic $ w(k,h)\mathcal{T}_i(m,k,h)$ exceeds a pre-defined threshold.
A weight function $w(\cdot,\cdot)$ is introduced to control the asymptotic size of the detection procedure. Typically $w(\cdot,\cdot)$ depends on the monitoring time $k$, and the window size $h$,
\begin{align}
  \label{eq:weight_function}
  w(k,h) =  \frac{1}{\sqrt{h}} \rho\left( \frac{k}{h} \right),
\end{align}
for some appropriate $\rho(\cdot)$. The choice of the weight function controls the sensitivity of the test. A wide range of weight functions can be used as long as they are continuous functions that satisfy $\inf_{0\leq t \leq T} \rho(t)> 0$. In this paper, we use the weight function proposed in \cite{10.2307/3533257} and \cite{doi:10.1002/jae.776}:
\begin{align*}
  \rho(t) = \max( 1, \log \left( 1 + t \right))^{-1/2}. 
\end{align*}
Intuitively, if there is no change the weighted MOSUM will remain small, but it will be large if there is a change. Figure \ref{fig:example_MOSUM} gives the behavior of weighted MOSUM statistic under the null and the alternative assumptions for one data stream.
\begin{figure}[h!]
  \begin{minipage}{.5\linewidth}
    \centering
    \subfloat[]{\label{main:a}\includegraphics[width=\textwidth]{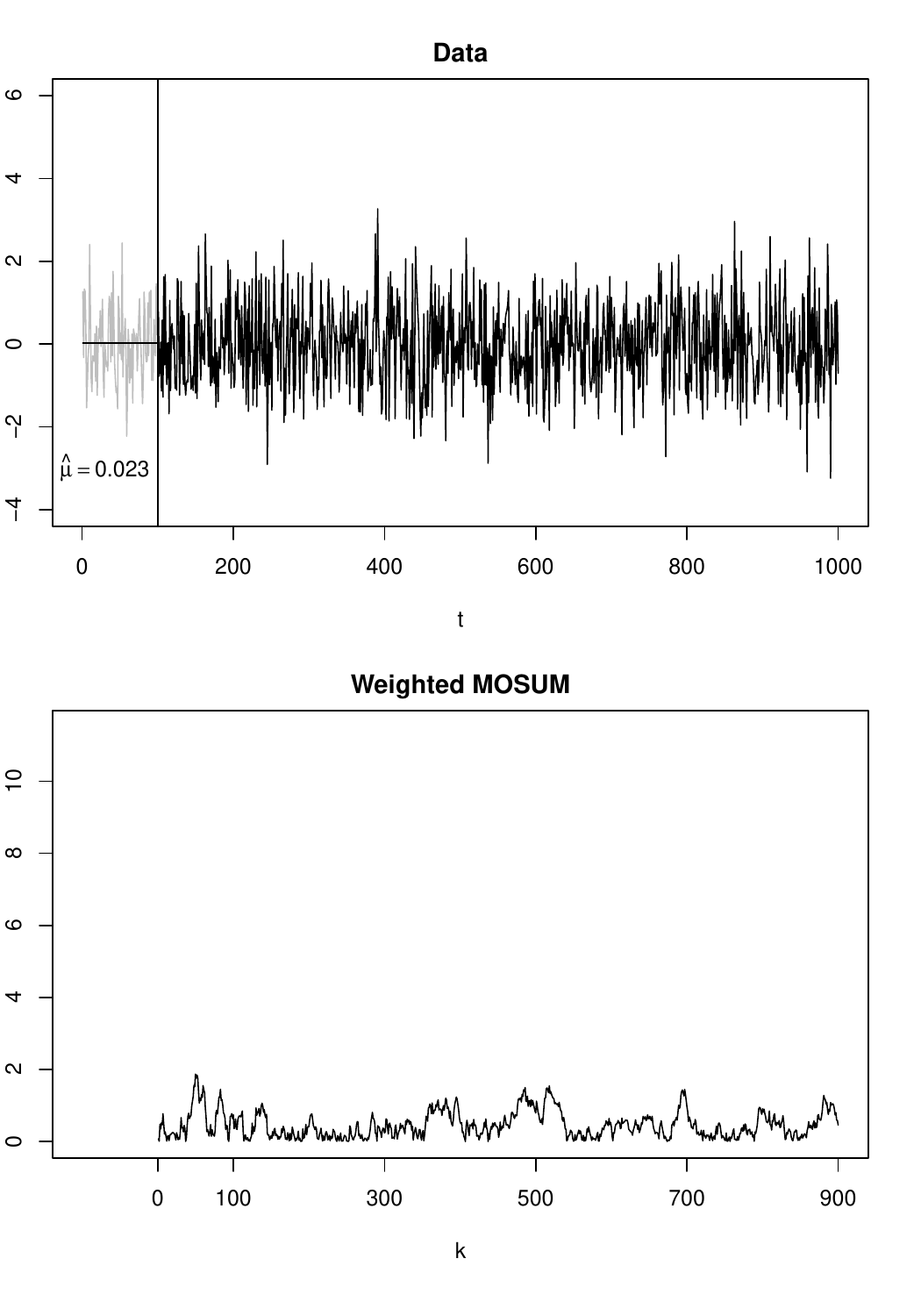}}
  \end{minipage}%
  \begin{minipage}{.5\linewidth}
    \centering
    \subfloat[]{\label{main:b}\includegraphics[width=\textwidth]{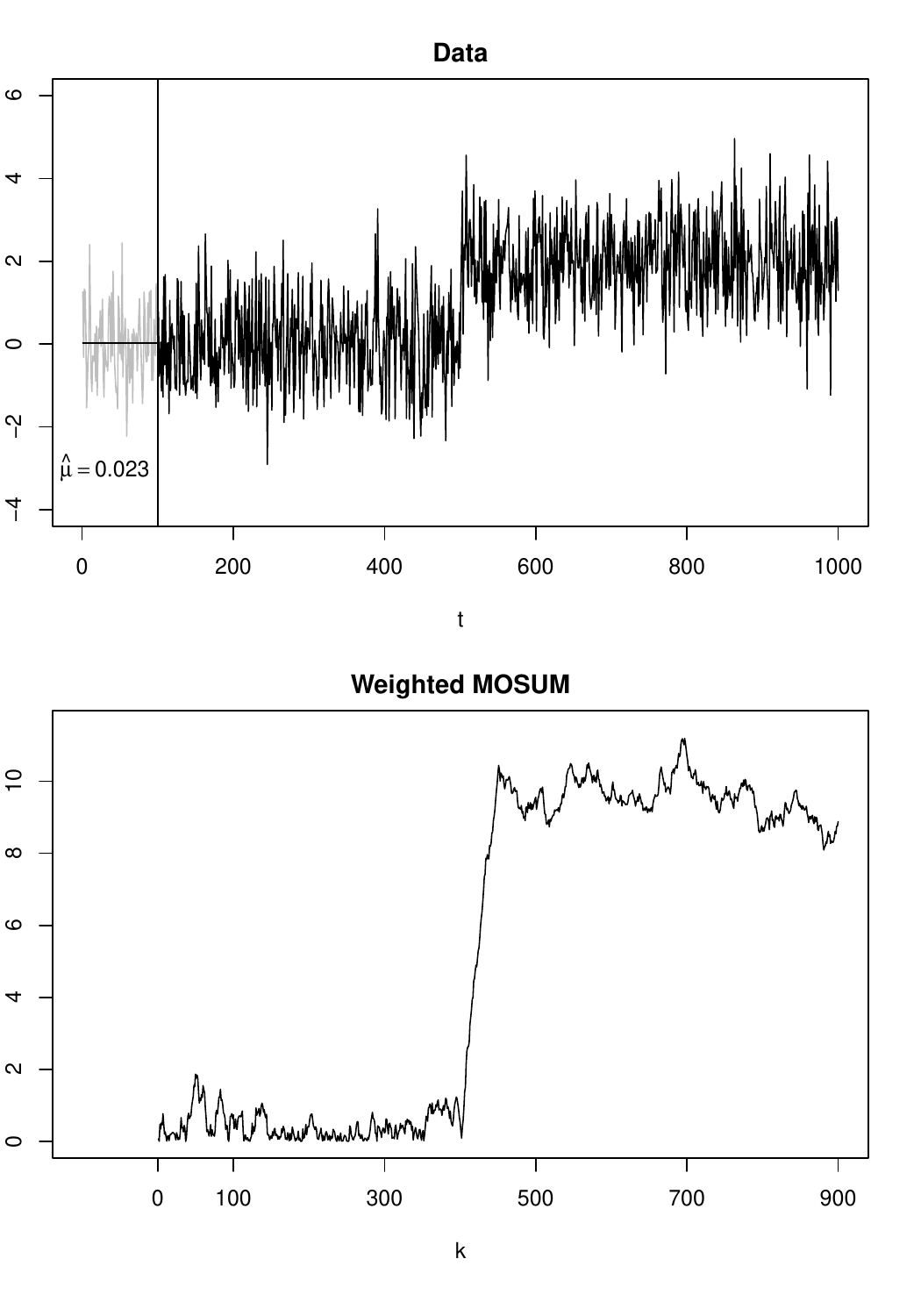}}
  \end{minipage}
  \caption{Example time series with no change (a) and a single change (b) in the top row. The bottom row shows the weighted MOSUM statistic with a historic period of length $m = 100$ and a window size of $h = 50$.}
  \label{fig:example_MOSUM}
\end{figure}
\subsection{Message passing}
The local monitoring described in the previous section is applied to each sensor independently. In order to make global decisions about the state of the system, messages from the sensors must be passed to the central hub (see Figure \ref{fig:schematic}). However, since there are constraints on communication in the system, the message passing process must be carefully designed.\\

At time $t = m + k$ , where $m$ is the historic period of length $m$, and $k$ is the monitoring time, each sensor makes a decision as to whether or not to transmit a message to the centre. This message vector is denoted as $\mathbf{M}_t = (M_{1,t}, M_{2,t}, \hdots , M_{d,t})$. We consider two different messaging regimes:
\begin{itemize}
    \item Centralized messaging regime: $\mathbf{M}_t =  \mathcal{T}_{i}(m,k,h)$.
    \item Distributed messaging regime: \begin{align}\label{distributed messaging regime}
  M_{i,t} =
  \begin{cases}
    \mathcal{T}_{i}(m,k,h) &\quad\text{if } w(k,h)\mathcal{T}_{i}(m,k,h) > c_{\textrm{Local}}, \\
    NULL &\quad\text{otherwise.} \\
  \end{cases}
\end{align}
\end{itemize}
The centralized massaging regime is one where there  is no constraint on the communication between the sensors and the centre, so all sensors send a message to the centre at each time instant. This is similar to the ``SUM'' scheme changepoint detection method proposed by \cite{10.1093/biomet/asq010}. However, when communication is expensive, a ``distributed'' messaging regime can be used where each of the sensors only send local monitoring statistics that exceed a chosen threshold. The $NULL$ means no message is sent. The threshold $c_{\textrm{Local}}$ can be chosen to control the fraction of transmitting sensors when there is no change. It is worth noting that when $c_{\textrm{Local}}=0$, the ``distributed'' messaging regime is equivalent to ``centralized'' messaging regime.
\subsection{Global monitoring}
In our paper, we assume that there is no communication delay between sensors and the central hub, so the message could be immediately received by the centre at time $t$. Based on the messages received, the centre will make the decision as to whether or not to flag a change.
\subsubsection{Combining messages}
Depending on different messaging regimes, the global MOSUM statistics are constructed as follows:
\begin{itemize}
    \item Centralized global MOSUM statistic: 
    \begin{align}
    \label{eq:global_MOSUM2} \mathcal{T}(m,k,h) = \sqrt{ \sum_{i=1}^{d} M_{i,t}^2   },
    \end{align}
    This is similar to the SUM scheme mentioned in Section \ref{sec:problem}. By using such a scheme, Formula \ref{eq:global_MOSUM2} is the idealistic scheme under dense change.
    \item Distributed global MOSUM statistic:
         \begin{align}
         \label{eq:global_MOSUM} \mathcal{T}(m,k,h) = \sqrt{ \sum_{i=1}^{d} M_{i,t}^2 \mathbbm{1}_{ \mathcal{T}_i(m,k,h) > c_{\textrm{Local}}}  },
          \end{align}
          where $NULL$ values in Formula \ref{distributed messaging regime} are taken to be zeros in the sum. The form of Equation \eqref{eq:global_MOSUM} is taken from the multivariate MOSUM \citep{kirch2015, Weber2017_1000068981, kirch2018}.
\end{itemize}
\subsubsection{Declaring the change}
Similar to the local monitoring procedure, a change is declared as soon as the weighted global MOSUM exceeds a threshold. A closed-end stopping rule can be used when the aim is to monitor changes within a fixed time. This can be formalised as
\begin{align}
  \label{eq:ST}
  \tau_{m,\Tilde{T}} = \min \left\{1 \leq k \leq \lfloor m\Tilde{T} \rfloor: w(k,h)\mathcal{T}(m,k,h) > c_{\textrm{Global}} \right\},
\end{align}
where $\min \{ \emptyset \} = \infty$ and the total length of the data $T=m\Tilde{T}$. If no change is detected by this stopping rule prior to $\lfloor m\Tilde{T} \rfloor$, the monitoring procedure is terminated. The parameter $\Tilde{T} > 0$ governing the length of the monitoring period is chosen in advance \citep{doi:10.1080/07474940801989087, AUE20122271}. \\\\
Figure \ref{fig:sparse_dense} shows the weighted global MOSUM statistic for the distributed and centralized messaging regimes on the same dataset. Whenever the weighted global MOSUM of distributed regime hits zero, there is no communication between the edges and the centre at that time.
\begin{figure}[]
  \centering
\includegraphics[width=0.7\textwidth]{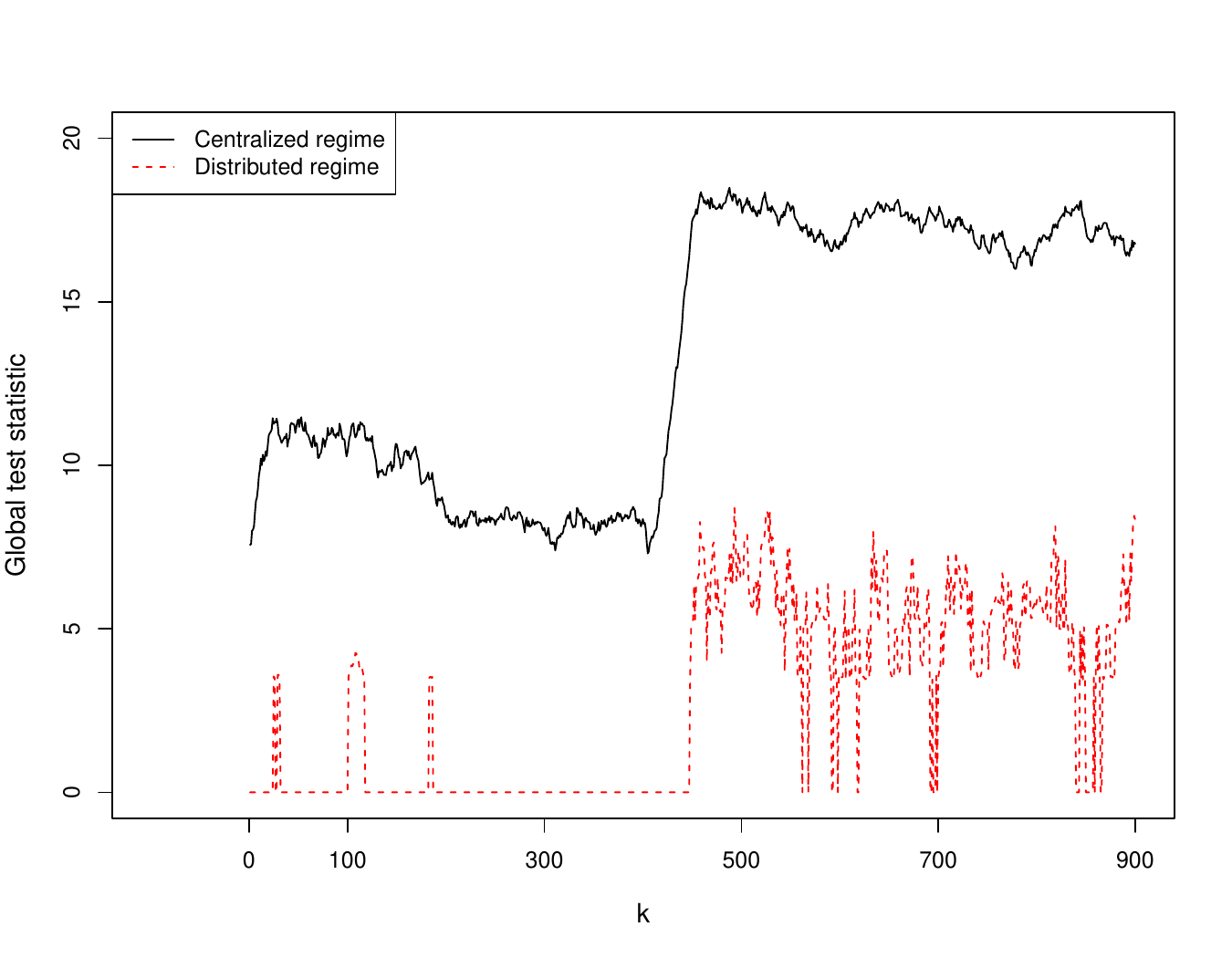}
  \caption{Example of the weighted global MOSUM statistic for the distributed (red dashed line) and centralized (black line) regime. The result is obtained with $T=1000, d=100, m=100$, $h=50, \delta=0.5$ and the number of affected sensors $p=50$. A value of $c_{\textrm{Local}} = 3.44$ was used in the distributed regime.}
  \label{fig:sparse_dense}
\end{figure}
In the next section, we will show the theoretical properties of our proposed method under $H_0$ and $H_A$.
\section{Theoretical properties for distributed MOSUM}
\label{sec:theory}
This section considers the theoretical properties of the closed-end stopping rule, $\tau_{m,\Tilde{T}}$ defined in Equation \eqref{eq:ST} as $m \rightarrow \infty$. Firstly, in Section \ref{sec:asymp_null} we find the limiting distribution under the null hypothesis for the different procedures. Then, appropriate choices for the thresholds, $c_{\textrm{Local}}$ and $c_{\textrm{Global}}$ are given in Section \ref{sec:theory_cvals} using these results. Finally, in Section \ref{sec:asymp_alt} we prove that the detection procedures we have studied are consistent under alternatives.\\

Three key assumptions are made in order to derive asymptotic results, which are the same in \cite{doi:10.1080/07474940801989087}, \cite{AUE20122271}, and \cite{Weber2017_1000068981}:
\begin{Assum}[Clean historic data]\label{clean}
  $h \rightarrow \infty$ as $m \rightarrow \infty$ and the location of the changepoint $\tau > m$ for $1 \leq i \leq d$.
\end{Assum}
This assumption is to guarantee we can get good estimators based on the training dataset, and it can be easily achieved in real applications.
\begin{Assum}[Asymptotic regime]\label{asymp}
  $h \rightarrow \infty$ as $m \rightarrow \infty$ and
  \begin{align*}
   \lim_{m \rightarrow \infty} \frac{h}{m} \rightarrow \beta \in (0,1]. 
  \end{align*}
\end{Assum}
This assumption quantifies the long run connection between the length of the historical period $m$ and the window size $h := h(m)$. 
\begin{Assum}[FCLT on errors]\label{FCLT}
  \begin{align*}
   \lim_{m \rightarrow \infty} \frac{1}{\sqrt{m}} S_i(mt) \overset{\mathcal{D}}{\longrightarrow} \sigma_i W_{i}(t) 
  \end{align*}
  where $\sigma_i > 0$, $\{ W_{i}(t), 0 \leq t < \infty\}$ is a standard Brownian motion when $h \rightarrow \infty$, and $S_i(x) = \sum_{t = 1}^{\lfloor x \rfloor} \epsilon_{i,t}$. $\sigma_i$ can be estimated by $\hat \sigma_i$. Furthermore, $\hat \sigma_i$ satisfying  $\hat \sigma_i \overset{\mathcal{P}}{\longrightarrow} \sigma_i$  as $m \rightarrow \infty$.
\end{Assum}
This assumption is a functional central limit theorem on the errors, $\epsilon$, in the model for the data \eqref{eq:model}.

\subsection{Asymptotics under the null}
\label{sec:asymp_null}
In this part, the asymptotic theories of our proposed method will be given, which can help guide the choice of thresholds.\\

The local monitoring process of our proposed method within each sensor is the same as univariate MOSUM detection process. Thus, Theorem \ref{thm:local} and Corollary \ref{thm:local_cor} of the local MOSUM can be directly cited from \cite{doi:10.1080/07474940801989087}, \cite{AUE20122271} and \cite{Weber2017_1000068981}. For simplicity, we denote
\begin{align}
\label{brownian}
    Z_i(t) = \left| W_{i}\left( \frac{1}{\beta} + t \right) - W_i\left( \frac{1}{\beta} + t - 1 \right) - \beta W_{i}\left( \frac{1}{\beta}\right) \right|, \quad 1 \leq i \leq d
\end{align}
where $\{ W_{i}(t), 0 \leq t < \infty\}$ are independent standard Brownian motions.
\begin{Thm}[\emph{Local MOSUM}]
  \label{thm:local} 
  If assumption \ref{clean}-\ref{FCLT}, and model \ref{eq:model_eq} holds, then under $H_0$, let $k = ht$ for any $t > 0$ 
  \begin{align*}
    \lim_{m \rightarrow \infty} w(k,h)\mathcal{T}_i(m,k,h) \overset{\mathcal{D}}{\longrightarrow} \rho(t)Z_i(t).
  \end{align*}
\end{Thm}
\begin{corollary}[\emph{Local MOSUM - asymptotic type-I error}]
  \label{thm:local_cor}
  Under $H_0$, for any $\Tilde{T} > 0$ and $i$th data stream,
  \begin{align*}
    \lim_{m \rightarrow \infty} P \left( \tau_{m,\Tilde{T}}^{(i)} < \infty \right) = P \left( \sup_{0 \leq t \leq \Tilde{T}/\beta} \rho(t)Z_{i}(t) > c_{\textrm{Local}} \right).
  \end{align*}
  Thus, the false alarm rate for one data stream is asymptotically equal to a pre-specified type-I-error $\in (0,1)$.
\end{corollary}
Following the results of local MOSUM, similar results for global MOSUM follow readily. These can be used to choose thresholds given the pre-defined Type-I-error. Below we obtain two limiting distributions, for the centralized and distributed regime settings of Section \ref{sec:methodology} respectively.
\begin{Thm}[\emph{Global MOSUM}]
  \label{thm:global_null}
  Let $k = ht$ for any $t > 0$, then under $H_0$,
  \begin{align*}
    \lim_{m \rightarrow \infty} w(k,h)\mathcal{T}(m,k,h) \overset{\mathcal{D}}{\longrightarrow}       \rho(t)
    \begin{cases}
      \sqrt{ \sum_{i=1}^{d} Z_{i}(t)^2  }    &\quad\textrm{centralized case,} \\
      \sqrt{ \sum_{i=1}^{d} Z_i(t)^2\mathbbm{1}_{ \rho(t)Z_{i}(t) > c_{\textrm{Local}}  }} &\quad\textrm{distributed case}.      
    \end{cases}
  \end{align*}
\end{Thm}
\begin{proof}
See Appendix \ref{sec:global_null}.
\end{proof}
Thus, their limiting distribution will be a function of Gaussian process. Using the Theorem \ref{thm:global_null}, the following may be obtained:
\begin{corollary}[\emph{Global MOSUM - asymptotic type-I error}]
  \label{thm:global_cor}
  Under $H_0$, for any $\Tilde{T} > 0$,
  \begin{align*}
    \lim_{m \rightarrow \infty} P \left( \tau_{m,\Tilde{T}} < \infty \right) = 
    \begin{cases} P\left(
     \sup_{0 \leq t \leq \Tilde{T}/\beta} \rho(t) \sqrt{ \sum_{i=1}^{d} Z_{i}(t)^2  }  > c_{\textrm{Global}} \right)   &\quad\textrm{centralized case,} \\
       P\left( \sup_{0 \leq t \leq \Tilde{T}/\beta} \rho(t) \sqrt{ \sum_{i=1}^{d} Z_i(t)^2\mathbbm{1}_{\rho(t)Z_{i}(t) > c_{\textrm{Local}}  }} > c_{\textrm{Global}} \right) &\quad\textrm{distributed case}.      
    \end{cases}
  \end{align*}
\end{corollary}
This result can lead us to find the local and global thresholds which can obtain the pre-defined type-I-error.
\subsection{Obtaining critical values}
\label{sec:theory_cvals}
Using the results of the previous section, appropriate critical values can be found such that the asymptotic type-I error is controlled for the different procedures. To achieve this the stochastic processes $\{ Z_{i}(t), 0 \leq t \leq \Tilde{T}/\beta, 1 \leq i \leq d\}$ need to be approximated on a fine grid. This is done in the same way as \cite{AUE20122271}, simulating the component standard Brownian motions using ten thousand i.i.d.~standard normal random variables. The parameters used were $\beta = 1/2$ and $\Tilde{T} = 10$. Tables \ref{tab:Local_globdense_cval} and \ref{tab:globsparse_cval} give critical values for $\alpha \in \{0.10, 0.05, 0.01\}$.
\begin{table}[h]
  \centering
  \begin{tabular}{ccc}
    \hline
    $\alpha$ & $c_{\textrm{Local}}$ & $c_{\textrm{Global}}^{\textrm{Centralized}}(\alpha)$   \\
    \hline
    0.10 & 0 & 14.1  \\
    0.05 & 0 & 14.4  \\
    0.01 & 0 & 15.0  \\
    \hline
  \end{tabular}  
  \caption{Critical values for the centralized procedures, results averaged over five thousand replications.}
  \label{tab:Local_globdense_cval}
\end{table}
\begin{table}[h]
  \centering
  \begin{tabular}{ccc}
    \hline
    $\alpha$                 & $c_{\textrm{Local}}$ & $c_{\textrm{Global}}^{\textrm{Distributed}}(\alpha)$ \\
    \hline
    \multirow{3}{*}{0.10}    &   3.15              &        7.48              \\
                             &   3.44              &        6.70            \\
                             &   4.05              &        5.20             \\
    \hline
    \multirow{3}{*}{0.05}    &   3.15              &        7.89         \\
                             &   3.44              &        7.16          \\
                             &   4.05              &        6.02          \\
    \hline
    \multirow{3}{*}{0.01}    &   3.15              &        8.74         \\
                             &   3.44              &        8.01          \\
                             &   4.05              &        6.59         \\
    \hline
  \end{tabular}
  \caption{Critical values for the distributed procedure with different values for $c_{\textrm{Local}}$, results averaged over five thousand replications.}
  \label{tab:globsparse_cval}
\end{table}\\
\noindent
Since the critical values obtained above are valid asymptotically (in $m$), an important question to consider is how they perform in finite samples. Numerical results of empirical size in the finite sample are shown in Table \ref{tab:FA}. Thse indicate that the implementation in the finite sample setting can be conservative,  as  per \cite{AUE20122271}. However, approximately, the type-I error is controlled at the correct level for both of the global procedures in finite samples. 

\begin{table}[ht]
  \centering
  \begin{tabular}{ccccccc}
    \hline
    \multirow{2}{*}{Method}   & \multirow{2}{*}{$c_{\textrm{Local}}$} & \multirow{2}{*}{$c_{\textrm{Global}}$} & \multicolumn{3}{c}{Proportion of false alarms}  \\
                              &                                      &                     & $m = 200$ & $m=400$ & $m=500$\\
    \hline
    \multirow{3}{*}{Distributed}   &        3.15         &        7.89      & 5.14\% &   5.18\% & 5.5\%           \\
                              &        3.44         &        7.16      & 5.3\% & 5.38\% &5.12\%  \\
                                &        4.05         &        6.02   & 5.5\%  & 5.64\%  &5.16\%   \\
      Centralized                     &         -           &        14.4     & 5.92\% &   5.28\%  & 5.3\%    \\
    \hline
  \end{tabular}
  \caption{Empirical size, results averaged over one thousand replications with $\alpha=0.05, \Tilde{T}=10$, and $\beta=1/2$.}
  \label{tab:FA}
\end{table}
\subsubsection{The choice of local threshold $c_{\textrm{Local}}$}
The values for $c_{\textrm{Local}}$ used in Table \ref{tab:globsparse_cval} are somewhat arbitrary. The main influence of the value of local threshold is that it controls the proportion of messages that the system can pass (on average) per iteration. For $d$ streams, the number of sensors passing message at each time step is: 
\begin{corollary}[\emph{Transmission cost}]
   For any t>0 and k=ht, the expected fraction of transmitting sensors at each time step is
  \begin{align*}
    \bar \Delta_t=d P\left(\rho(t)|Z| >c_{\textrm{Local}}\right).
  \end{align*}
where $Z$ is the standard normal distribution. 
\end{corollary}
Therefore, the local threshold can be chosen based on the restriction of the transmission cost. Combined with pre-defined type-I-error, the global threshold will be given based on Theorem \ref{thm:global_null}.

\subsection{Asymptotics under the alternative}
\label{sec:asymp_alt}

Under the alternative it is assumed that there is a changepoint at monitoring time $k^{*}$ and a subset $\mathcal{S}$ of the data streams have an altered mean
\begin{align*}
  H_A: \tau = m + k^{*} \quad  \& \quad \exists \mathcal{S} \subset \{1, 2, \hdots, d\}: \delta_i \ne 0 \quad \textrm{for } i \in \mathcal{S}.
\end{align*}
Deriving sharp asymptotic results on the detection delay of the proposed method is challenging, and thus we focus only on giving consistency results. A procedure is consistent if it stops in finite time with probability approaching one as $m \rightarrow \infty$. In other words, the test statistic should tend to infinity as $m \rightarrow \infty$. In the asymptotic regime of interest, we additionally assume that the changepoint $k^{*}$ grows at the same order as $h$, that is $\frac{k^*}{h} \rightarrow  \gamma \geq 0$,
and the size of change $\delta_{i,t}$ satisfies $\sqrt{h}|\delta_{i,t}|\rightarrow \infty$ as $m \rightarrow \infty$ and $h \rightarrow \infty$. These assumptions are the same in \cite{AUE20122271}.
\begin{Thm}[\emph{Global MOSUM: Consistency}]
  \label{thm:global_alt}
  If the assumption above holds, under $H_A$,
  \begin{enumerate}[label=(\roman*)]
  \item the changepoint $k^{*} \leq \lfloor h \nu \rfloor$ for some $0 < \nu < \Tilde{T}\frac{m}{h}$,
  \item there exists a constant $c > 0$ such that $\rho(x+1) \geq c$ for all $x \in (\nu,  \Tilde{T}\frac{m}{h} - 1)$.
  \end{enumerate}
  Then, as $m \rightarrow \infty$ and $h \rightarrow \infty$
  \begin{align*}
    \max_{1 \leq k \leq \lfloor m\Tilde{T} \rfloor} w(k,h)\mathcal{T}(m,k,h) \overset{\mathcal{P}}{\longrightarrow} \infty.
  \end{align*}
\end{Thm}
\begin{proof}
See Appendix \ref{sec:global_alt}.
\end{proof}
Thus, our proposed method is consistent.
\section{Simulations}
\label{sec:sims}
In this section, we will present the numeric performance of our algorithm. Since the SUM procedure that is optimal when the change is dense, we will evaluate the performance in the dense case, specifically when the affected data streams $p=d$. Firstly, the different practical choices of thresholds at fixed type-I-error will be investigated. Here the performance of our proposed method was also compared against the idealistic setting. Finally, the effect of parameters and the violation of the independence assumption are investigated.\\\\
The set-up of the simulations is as follows. For simplicity, the data generating process under the null is that $X_{i,t} \sim N(0,1)$ for $1 \leq i \leq d$ and $1 \leq t \leq T$. To compare fairly, the type-I-error of all procedures is controlled to be $0.05$ under the null.\\  

The family of alternatives considered is that
\begin{align*}
      X_{i,t} \sim N(0,1)\quad \textrm{for} \quad 1 \leq i \leq d, 1 \leq t < \tau \quad \textrm{and} \quad X_{i,t} \sim N(\delta_i,1) \quad \textrm{for} \quad 1 \leq i \leq d, \tau \leq t \leq T.
\end{align*}
We assume the change will affect all the sensors instantaneously. But the size of the change is unknown. We consider two scenarios of mean shift: 1) Same size: $\delta_i=\delta=$ some constant values for $1 \leq i \leq d $; 2) Random size: $\delta_i=\eta N(0,1)$, where $\eta$ is the scale factor controlling the magnitude of size.
The average detection delay (ADD) $\bar{D}$ and average communication cost $\bar{\Delta}$ are then measured:
\begin{align*}
    \bar{D} &= E(\hat \tau -\tau |\hat \tau >\tau)\\
     \bar{\Delta} &=\sum_{t=1}^{\hat \tau}\frac{\sum_{i=1}^d\mathbbm{l}(w(k,h)\mathcal{T}_i(m,k,h)>c_{\textrm{Local}})}{\hat \tau -m +1},
\end{align*}
\subsection{The numerical dependency on local thresholds}\label{practical choice}
Our proposed method requires specifying two thresholds. Usually, $c_{\textrm{Global}}$ can be given based on the Theorem \ref{thm:global_cor} once $\alpha$ and $c_{\textrm{Local}}$ are confirmed. Therefore, it is crucial to pick an appropriate local threshold. This section gives numeric results with different values of local thresholds, which may provide some guidance in choosing the local threshold.

\begin{figure}[]%
 \centering
 \subfloat[Transmission cost for $d=100$ data streams.]{\label{trans_mu}\includegraphics[width=0.49\textwidth]{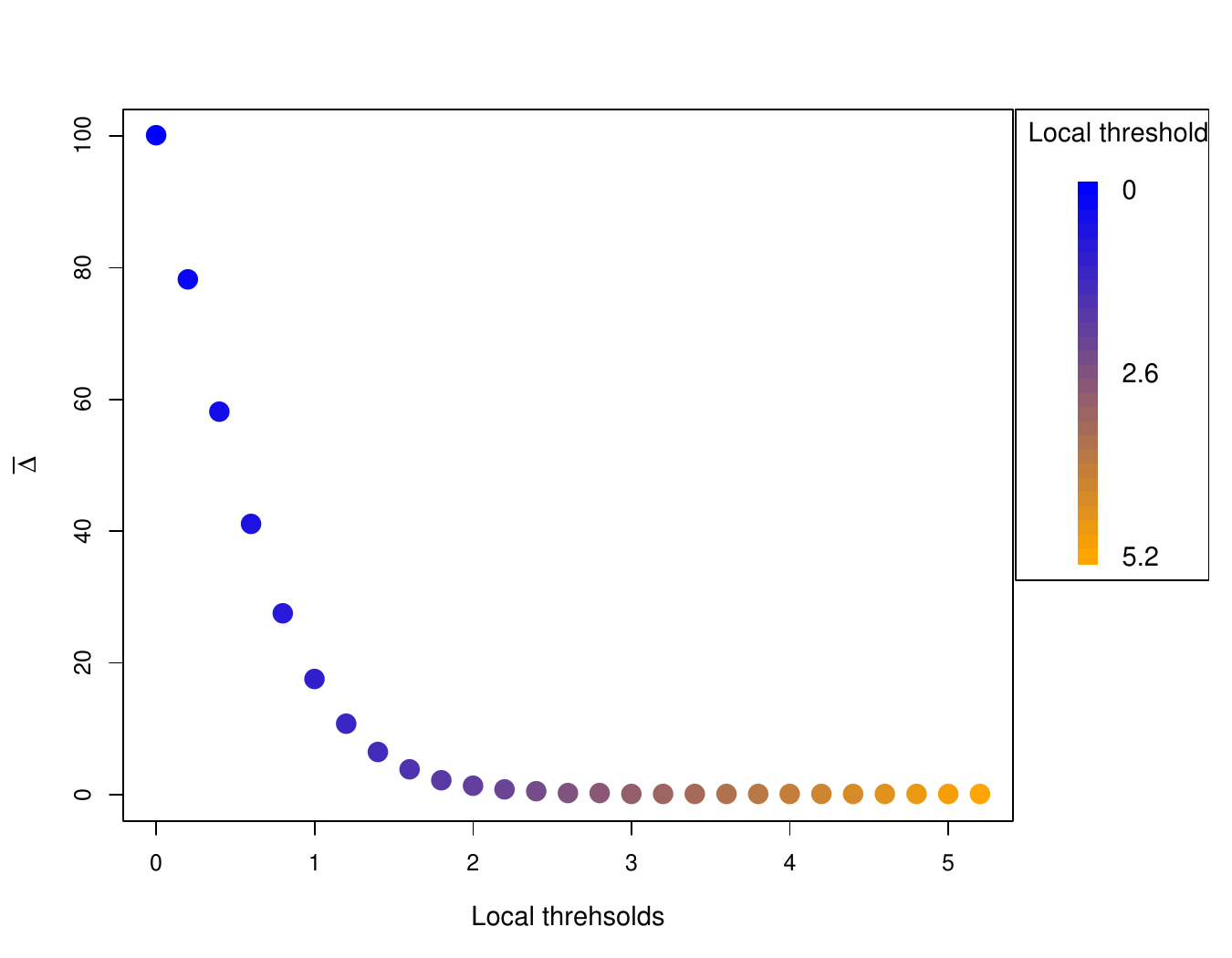}}\\
 \subfloat[average detection delay versus fixed $\delta$.]{\includegraphics[width=0.49\textwidth]{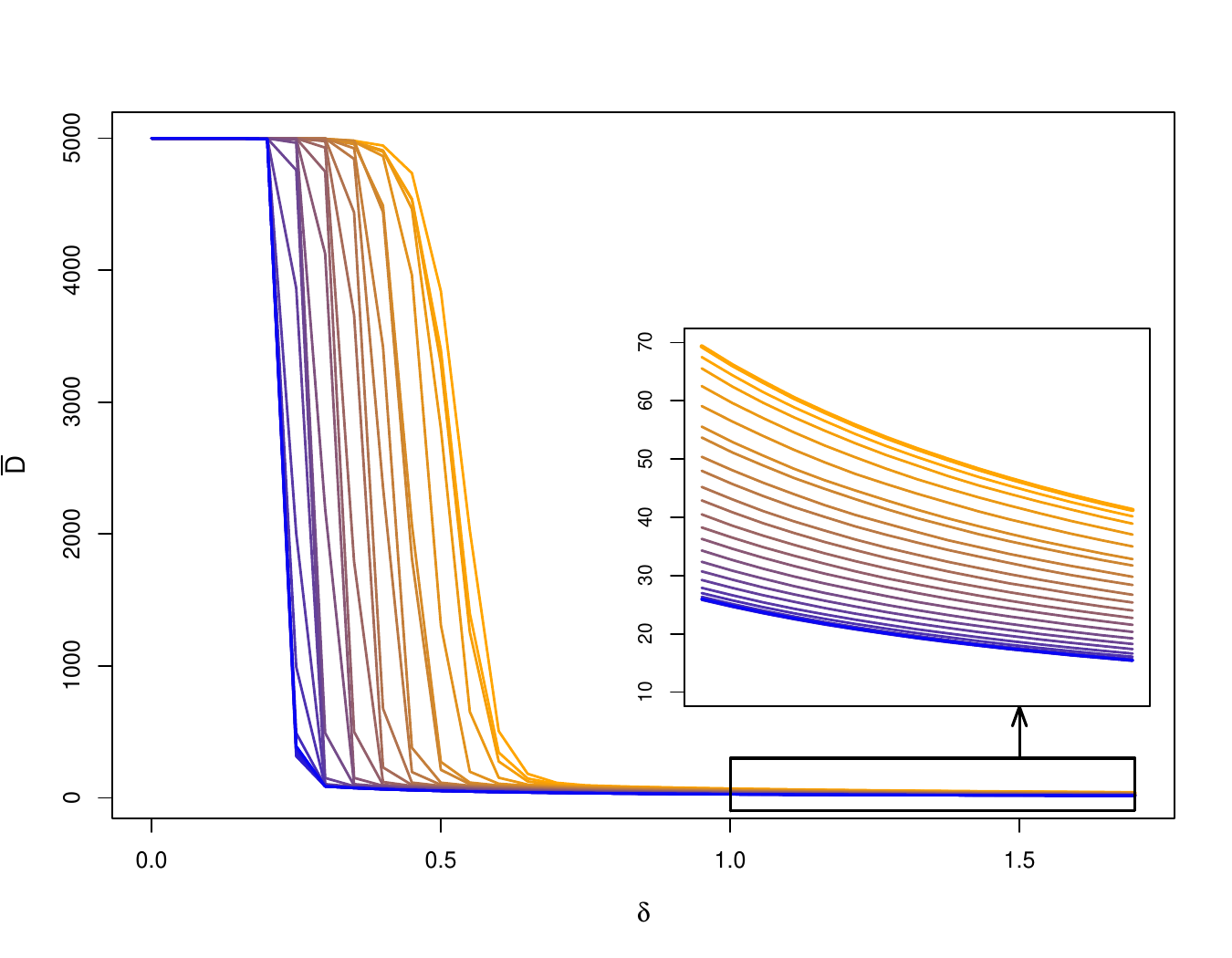}}
  \subfloat[average detection delay versus $\eta$.]{\includegraphics[width=0.49\textwidth]{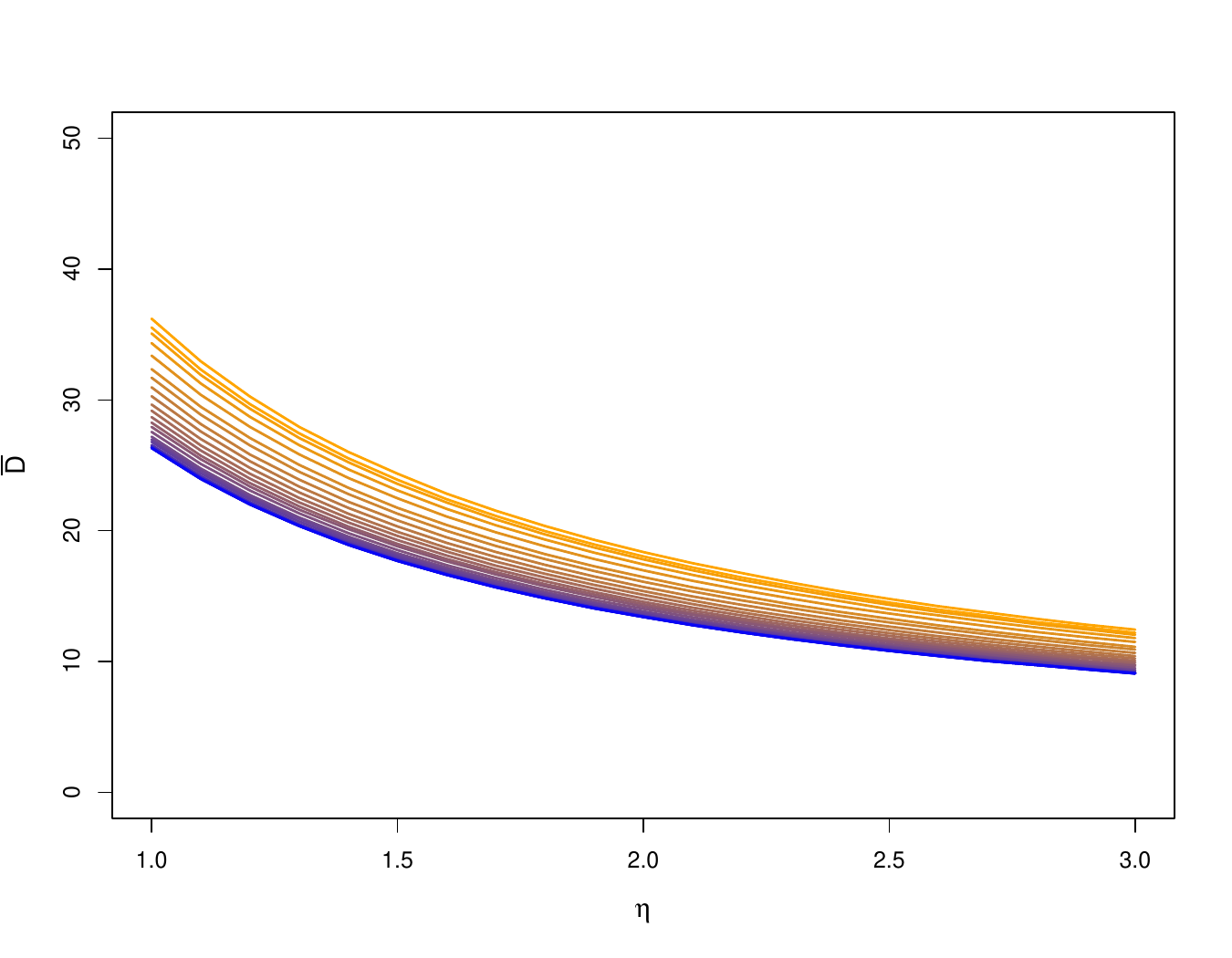}}
 \caption{The average number of messages transmitted to the centre (top) and average detection delay across varying mean shifts (bottom). Results are obtained when $m = 200$, $h = 100$, $T=10000, \tau = 5000, \alpha=0.05$. Each line corresponds to a different local threshold, which is labelled on the top right. The colour changes from orange to blue as the local thresholds increase from 0 to 5.2. When the local threshold is 5.2, the global threshold will be 0. So all possible combinations of thresholds are covered.} 
 \label{local_thresholds_mu}
\end{figure}

\noindent Figure \ref{local_thresholds_mu} gives the average detection delay and transmission cost for different values of local thresholds. There is a trade-off between communication savings and detection performance when choosing the local threshold. Larger local thresholds can reduce the transmission cost but will also lead to longer delays, especially when the change is small. However, with the increase in the mean shift, the detecting power of larger thresholds will close to that of small thresholds.\\

A centralized framework can be seen as an idealistic setting, which is equivalent to distributed setting when $c_{\textrm{local}}=0$. Compared with the idealistic setting, the distributed MOSUM can achieve similar performance when the size of the change is not small but also reduces massive transmission costs. But we will lose power in detecting small changes. We show the result below that distributed MOSUM can approximate the performance of idealistic setting overall by increasing the window size.

\subsection{The numerical dependency on parameters}
One advantage of using MOSUM statistics is that we do not need to specify the post-change mean. Instead, our proposed method requires specifying the window size $h$ and the training size $m$. In this section, we will investigate the impact of bandwidth and training size.
\subsubsection{The impact of bandwidth}\label{sec:bandwidth}
As shown in Figure \ref{5000a}, increasing the window size can increase the power of detecting small changes while leading to a slight delay in detecting large changes. Although increasing window size will increase the storage cost, it will not significantly increase the transmission cost as shown in Figure \ref{transbandwidth}. This drive us to think about whether we can improve the ability of distributed MOSUM with a large threshold to detect small changes by increasing the window size. Ideally, we would expect distributed MOSUM with increased window size can achieve similar performance as the idealistic setting.
\begin{figure}[H]%
 \centering
 \subfloat[$\bar D$ versus $\delta$ for $h=20,50,100$.]{\includegraphics[width=0.49\textwidth]{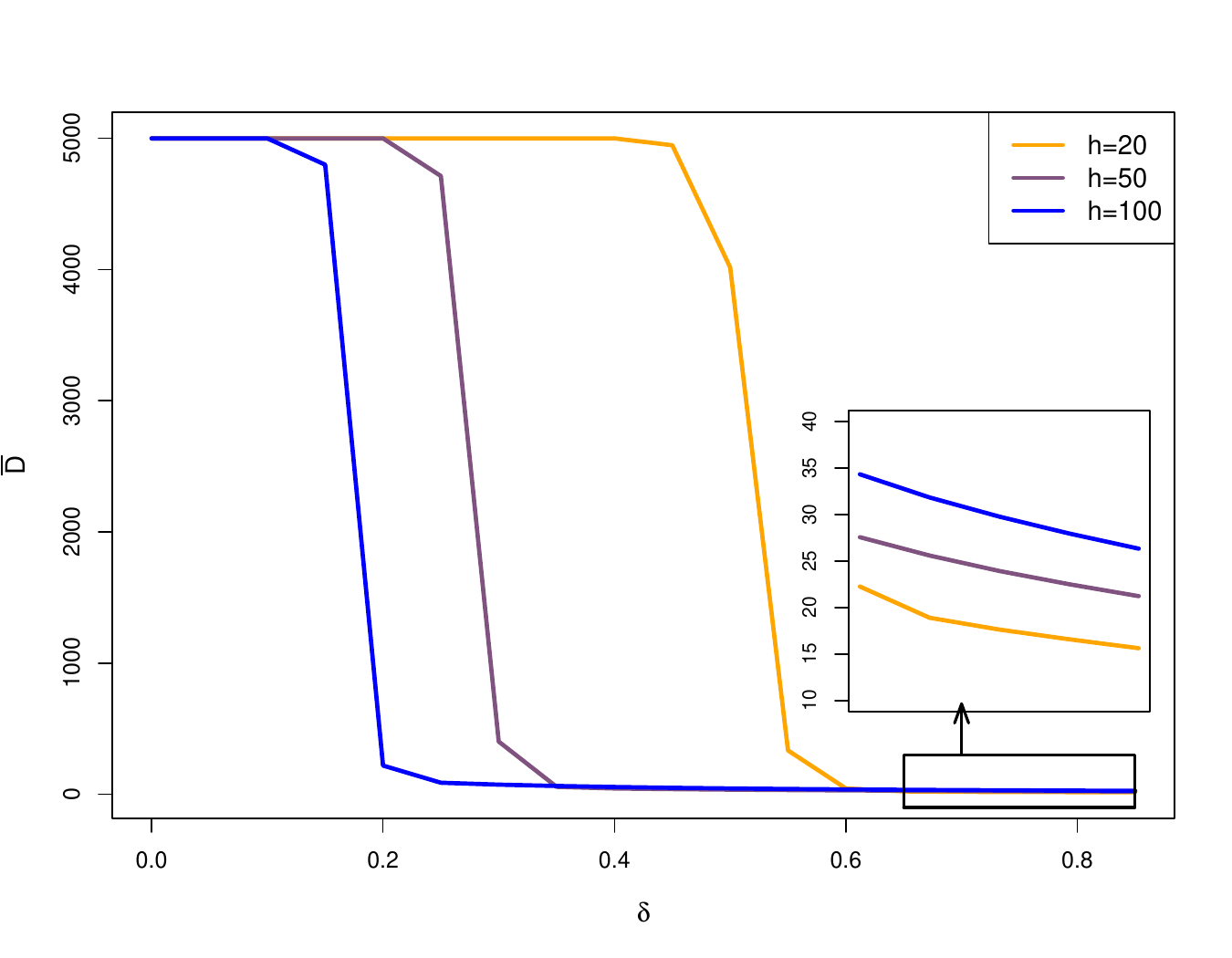}\label{5000a}}
 \subfloat[$\bar\Delta$ versus $h$.]{\includegraphics[width=0.49\textwidth]{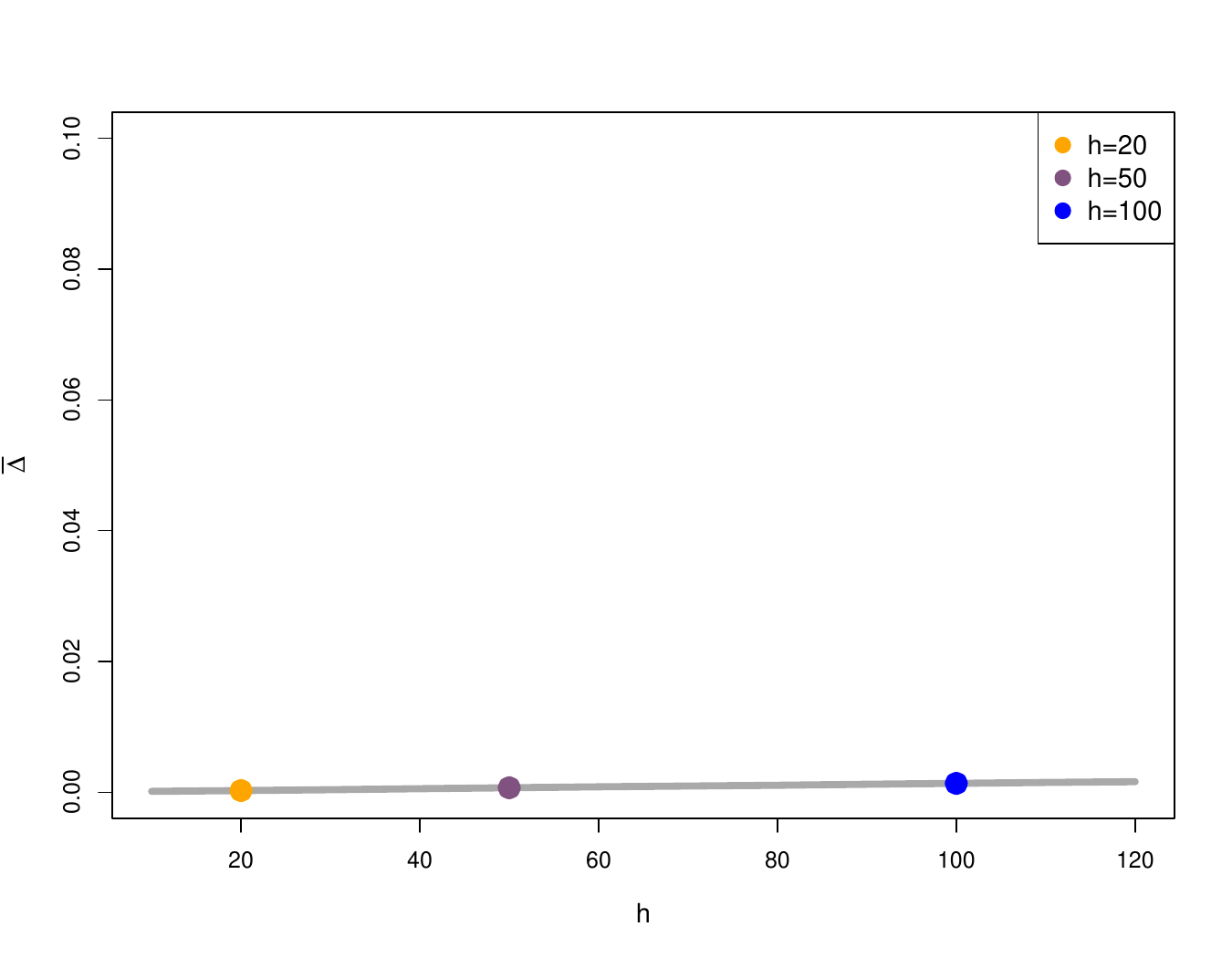}\label{transbandwidth}}%
 \caption{The influence of window size. Results are obtained over 1000 replications and take $m=200, d=100, T=10000, \tau=5000, \alpha=0.05, c_{\textrm{Local}}=3.44$.}
 \label{bandwidth_affect}%
\end{figure}
 \textbf{Recovering detectability}\\
 
For simplicity, we denote that the default window size for centralized MOSUM is $h^0$, and $h^*$ is the smallest window size that would allow distributed MOSUM to have similar performance as the idealistic setting. It is difficult to develop a neat theoretical formula between $h^*$ and $h^0$. But we can approximately find $h^*$ under alternatives by simulation. Our idea can be summarized as follows, and Figure \ref{increase_bandwidth} is the graphic explanation:
\begin{itemize}
    \item The behaviour of $\bar D$ will decrease dramatically when the mean shift is within a certain range (gray area). Therefore, we can find the median or mean $\delta$ of this certain range, denoted by $\delta^0$. Also, the corresponding ADD  $\bar D^0$ can be calculated. 
    \item Fix $\delta^0$, calculate $\bar D ^ {c_{\textrm{Local}}}(h)$ iteratively for distributed MOSUM, where $h\in[h^0,m]$. 
    \item The optimal window size $h^*=\arg \min \left\{\bar D ^ {c_{\textrm{Local}}}(h)-\bar D^0 (h^0) \right\}$. See blue arrow ($h^*$) is shorter than yellow arrow (h).
\end{itemize}
\begin{figure}[H]%
 \centering
 \includegraphics[width=0.49\textwidth]{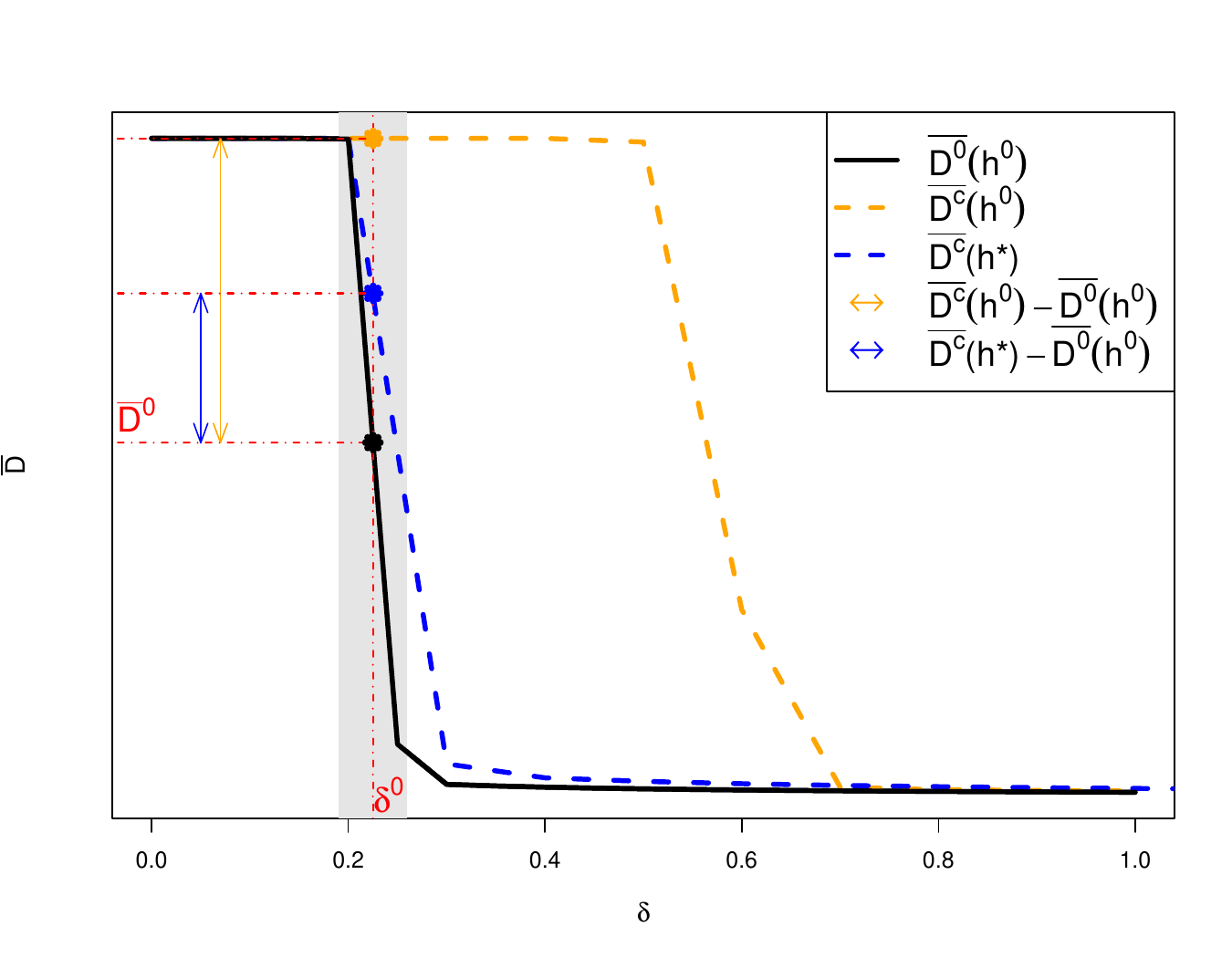}
 \caption{An graphic explanation of our proposed idea. Black line is the ADD for centralized MOSUM with window size $h$. Yellow dashed line is the ADD for distributed MOSUM with window size $h$; while blue dashed is the ADD for distributed MOSUM with window size $h*$.}
 \label{increase_bandwidth}%
\end{figure}
Figure \ref{increase_h_result} displays the simulation results that, for distributed regime, we can recover the same detectability of the centralized statistic by inflating $h$. \\
\begin{figure}[]%
 \centering
 \subfloat[$h^*$ found by our idea]{\includegraphics[width=0.49\textwidth]{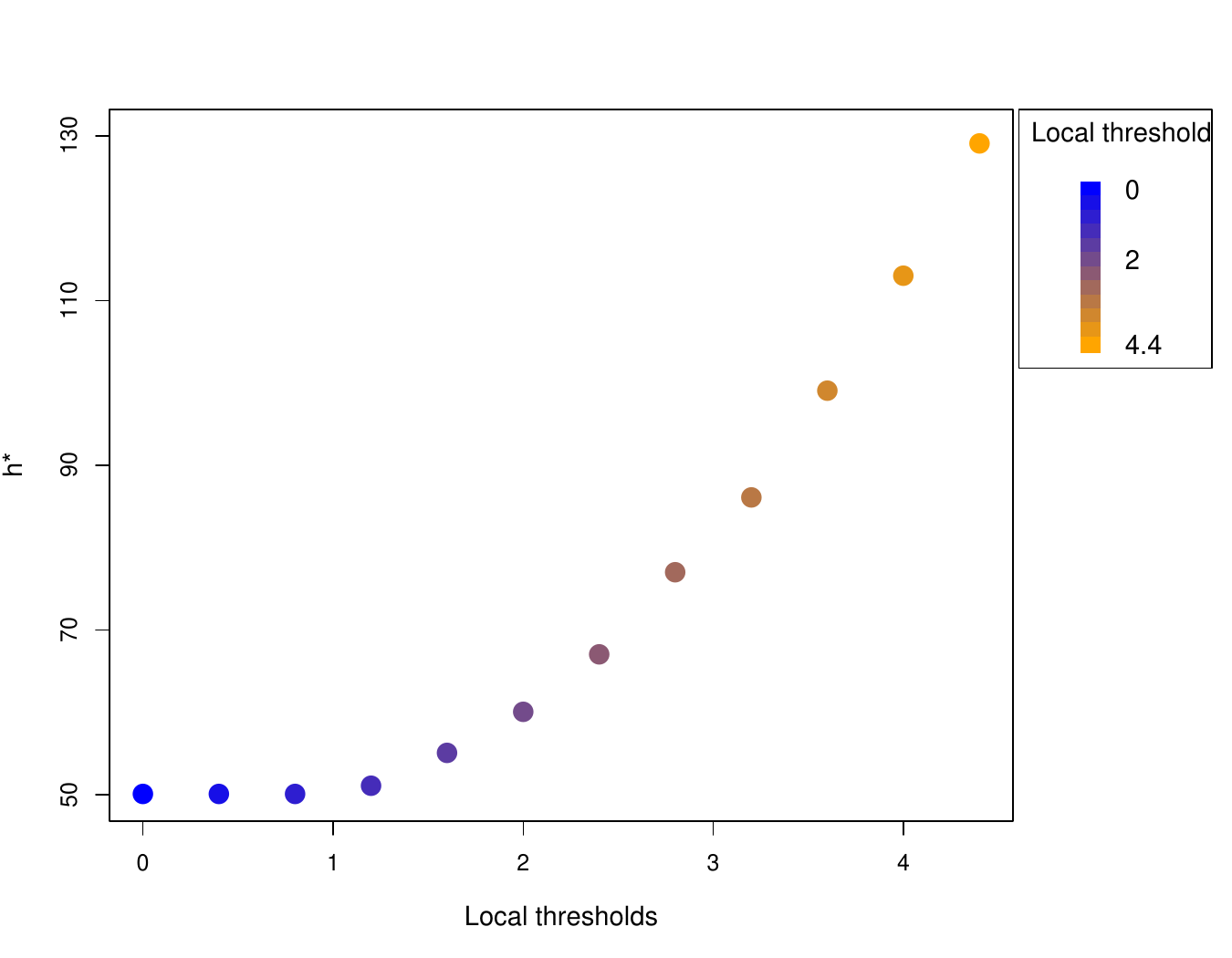}}
  \subfloat[average detection delay for distributed MOSUM (fixed $h=50$, gray line) and distributed MOSUM ($h=h^*$, colored line). Different lines corresponds to different $c_{\textrm{Local}}$]{\includegraphics[width=0.49\textwidth]{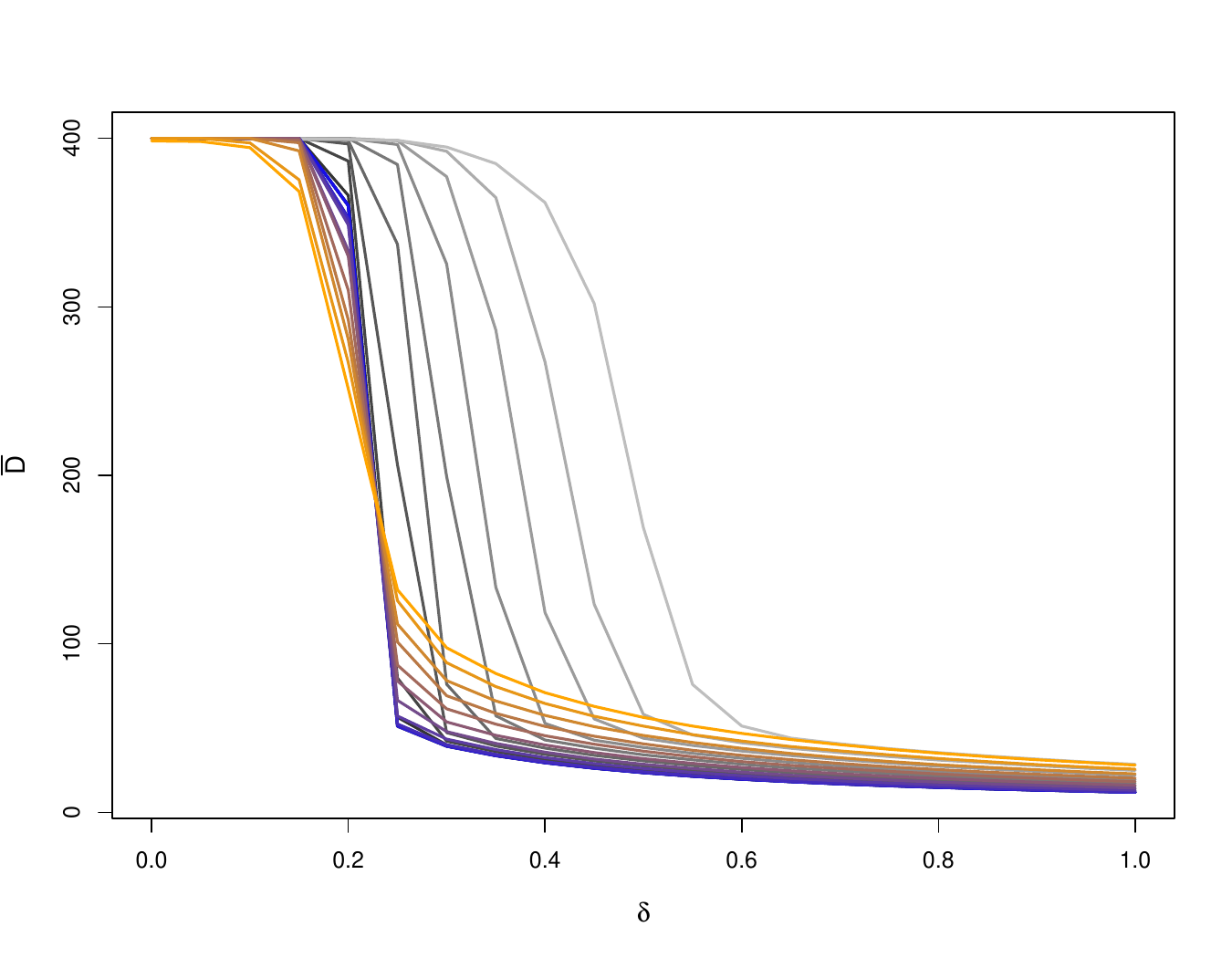}}
 \caption{An simple example showing that distributed MOSUM could approximate the detection power of centralized MOSUM by inflating window size. Results are obtained over 500 replications and take $m=200, d=100, T=1000, \tau=600, \alpha=0.05$, and $c_{\textrm{Local}}\in [0,4.4]$. When $c_{\textrm{Local}} = 4.4$, $c_{\textrm{Global}} = 0$. So all possible local thresholds are covered. For centralized setting, window size $h^0=50$.} 
 \label{increase_h_result}
\end{figure}
\subsubsection{The impact of the training dataset}
Fix the bandwidth $h$, the impact of the size of the training dataset can be investigated. Table \ref{tab:Training size} gives the thresholds, empirical size, and mean square errors (MSE) of estimated baseline parameters in our simulation. As we expected, the larger the training size is, the more accurate estimators are. Figure \ref{trainingsize} indicates that overall the detection powers of four different sizes of training datasets are similar. A larger training size could slightly increase the detection power when detecting small changes, which is attributed to more accurate estimators. Thus, in the real application, it is beneficial to choose a large-size training dataset because it is not expensive that can be done offline.
\begin{table}[]
\centering
\begin{tabular}{c|cccc}
\hline
             & $m=80$      & $m=100$       & $m=500$ & $m=1000$     \\ \hline
$c_{\textrm{Global}}$   & 9.039   &8.159     &6.014 & 5.708         \\ \hline
Empirical size & 0.007 & 0.007 & 0.009 & 0.006\\\hline
MSE for mean & 0.0125    & 0.01      & 0.002    &0.001\\\hline
MSE for sd   & 0.006    & 0.005     & 0.0001    & 0.0001\\ \hline
\end{tabular}
\caption{Empirical size, and MSE for estimated mean and standard deviation, results averaged over one thousand replications with $c_{\textrm{Local}}=3.44, h=50, T=6000$ and $\alpha=0.05$.}
\label{tab:Training size}
\end{table}
\begin{figure}[H]%
 \centering
 \includegraphics[width=0.5\textwidth]{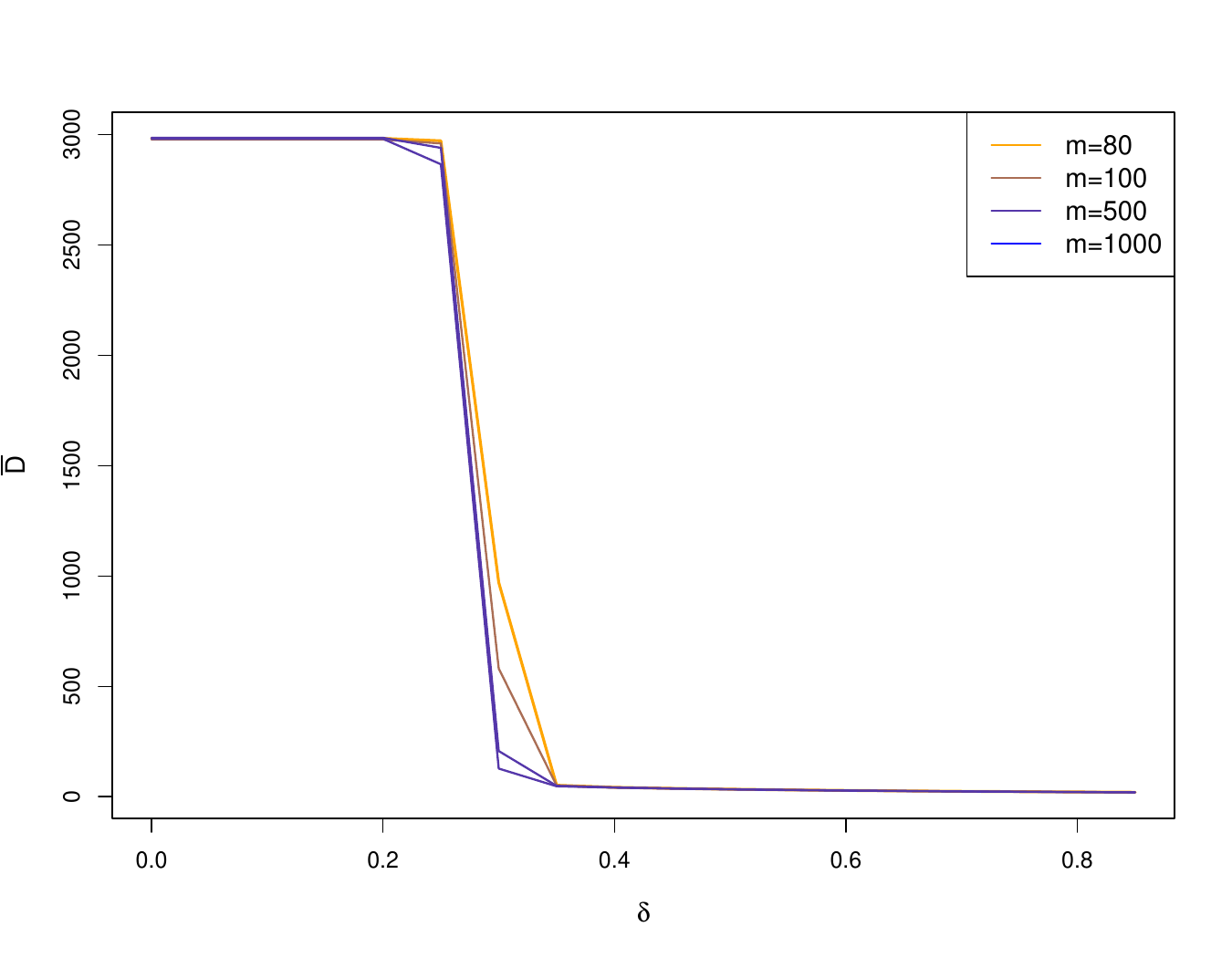}
 \caption{$\bar D$ versus $\delta$ when varing the size of training dataset. Result averaged over 500 replications with $\alpha=0.05, c_{\textrm{Local}}=3.44, T=6000$, $\tau=3000$ and $h=50$. The corresponding global thresholds are shown in Table \ref{tab:Training size}.}
 \label{trainingsize}%
\end{figure}
\subsection{The violation of the independence assumption}\label{sec:sim-ind}
Before, we assume that there is temporal independence among observations. However, it may not always hold in the real application. This section will investigate the performance when this assumption is violated. Here we measure our algorithm under AR(1) noise process, that is
\begin{align*}
    X_{i,t}= \delta_{i,t}\mathbbm{1}_{\left \{t>\tau \right \}} +\epsilon_{i,t},
\end{align*}
where $\epsilon_{i,t} = \phi \epsilon_{i,t-1}+v_t$ with $v_t \sim N(0,1)$. $|\phi|<1$ is used to measure the strength of the auto-correlation.\\

Auto-correlation will inflate the variance of data. There are two possible ways to handle this problem. The first one is to estimate the long-run variance as shown in Section \ref{kernel}. And one can also inflate the thresholds. We measure the false positives, average detection delay and the number of transmitted messages with fixed type-I-error of these two solutions under different scenarios. For better comparison, we also show the result of MOSUM without any adjustment. This will give us hints that to what extent our method fails to detect the change if we ignore the auto-correlation.\\

As Table \ref{autocorrelation} shows, our proposed method without adjustment can lose the ability to detect changes when introducing auto-correlation, that it fails to detect the change and always alarms. The performances of MOSUM with inflating thresholds are generally better than MOSUM with LRV since the former can detect the change in most scenarios. However, for those scenarios that the MOSUM with LRV can detect (usually $\delta$ is not small and $\phi$ is not large), it always has the lowest transmission cost and reasonable detection power. For example, when $p=100, \delta=1$, and $\phi=0.25$, both solutions have similar false positive rates and average detection delay, while MOSUM with LRV has lower transmission cost. It is surprising that estimating LRV has the lowest false positive rates and average detection delay when $\phi=0$ and $p=100/50$. This may be because it underestimates the variance.\\
 
However, when the auto-correlation is serious, it is not appropriate to apply our method to the raw data. Instead, it is more reasonable to apply our method after pre-processing the data, such as the residuals of AR models.
\begin{landscape}
\begin{table}[]\scalebox{0.8}{
\begin{tabular}{c|c|c|c|c|c|c|c|c|c|c|c|c|c|c|c|c}
\hline
        \multirow{3}{*}{$\phi$} & \multirow{3}{*}{Methods}  & \multicolumn{5}{c|}{$p=100$}  &\multicolumn{5}{c|}{$p=50$}  & \multicolumn{5}{c}{$p=10$}  \\ \cline{3-17}
        ~ & ~ &  ~ & \multicolumn{2}{c|}{$\delta=0.5$} &  \multicolumn{2}{c|}{$\delta=1$} & ~ &  \multicolumn{2}{c|}{$\delta=0.5$}&  \multicolumn{2}{c|}{$\delta=1$} & ~ &  \multicolumn{2}{c|}{$\delta=0.5$} &  \multicolumn{2}{c}{$\delta=1$}  \\ \cline{3-17}
        ~ & ~  & FP & D & Trans & D & Trans & FP & D & Trans & D & Trans & FP & D & Trans & D & Trans \\ \hline
        \multirow{3}{*}{0} & constant  & \blue{0.30} & \blue{90.67} & \blue{0.03} & \blue{45.75} & \blue{0.02} & \blue{0.29} & \blue{116.68} & \blue{0.03} & \blue{50.07} & \blue{0.03} & \blue{0.23} & 3767.69 & 0.18 & \blue{63.09} & \blue{0.03} \\ \cline{2-17}
         & inflating & \blue{0.30} & \blue{90.67} & \blue{0.03} & \blue{45.75} & \blue{0.02} & \blue{0.29} & \blue{116.68} & \blue{0.03} & \blue{50.07} & \blue{0.03} & \blue{0.23} & 3767.69 & 0.18 & \blue{63.09} & \blue{0.03} \\ \cline{2-17}
         & LRV  & \blue{0.19} & \blue{85.85} & \blue{0.05} & \blue{43.41} & \blue{0.05} & \blue{0.28} & \blue{105.74} & \blue{0.06} & \blue{48.42} & \blue{0.05} & \blue{0.32} & 4198.59 & 0.34 & \blue{64.32} & \blue{0.05} \\ \hline
        \multirow{3}{*}{0.25} & constant  & 63.49 & 79.82 & 0.19 & 40.36 & 0.18 & 62.04 & 92.62 & 0.19 & 45.90 & 0.18 & 61.95 & 2059.06 & 0.20 & 62.85 & 0.18 \\ \cline{2-17}
        & inflating  & \blue{0.33} & \blue{104.78} & \blue{0.22} & \blue{49.04} & \blue{0.20} & \blue{0.21} & 1413.09 & 0.76 & \blue{55.98} & \blue{0.20} & \blue{0.40} & 5000.00 & 0.48 & \blue{81.90} & \blue{0.21} \\ \cline{2-17}
         & LRV  & \blue{0.30} & 1294.64 & 0.21 & \blue{57.21} & \blue{0.05} & \blue{0.21} & 4228.80 & 0.30 & \blue{63.95} & \blue{0.05} & \blue{0.22} & 4995.26 & 0.09 & \blue{85.23} & \blue{0.05} \\ \hline
        \multirow{3}{*}{0.5} & constant  & 499.42 & 64.17 & 1.20 & 33.07 & 1.20 & 507.02 & 78.94 & 1.21 & 40.47 & 1.21 & 503.93 & 764.67 & 1.14 & 64.74 & 1.20 \\ \cline{2-17}
         & inflating  & \blue{0.57} & 4663.87 & 5.15 & \blue{60.51} & \blue{1.26} & \blue{0.34} & 5000.00 & 3.06 & \blue{73.66} & \blue{1.28} & 0.59 & 5000.00 & 1.10 & 5000.00 & 4.28 \\ \cline{2-17}
         & LRV & \blue{0.43} & 4996.68 & 0.09 & \blue{83.57} & \blue{0.05} & \blue{0.38} & 5000.00 & 0.06 & \blue{97.18} & \blue{0.06} & 0.59 & 5000.00 & 0.03 & 3878.04 & 0.35 \\ \hline
        \multirow{3}{*}{0.75} & constant  & 4219.42 & 5.74 & 8.35 & 3.99 & 8.35 & 4225.67 & 7.56 & 8.35 & 5.35 & 8.35 & 4228.56 & 8.25 & 8.37 & 6.94 & 8.37 \\ \cline{2-17}
        ~ & inflating  & \blue{0.76} & 5000.00 & 10.19 & \blue{498.83} & \blue{10.43} & 0.85 & 5000.00 & 7.84 & 5000.00 & 15.63 & 0.69 & 5000.00 & 6.02 & 5000.00 & 7.59 \\ \cline{2-17}
        ~ & LRV  & 0.64 & 5000.00 & 0.05 & 4980.28 & 0.25 & 0.59 & 5000.00 & 0.04 & 5000.00 & 0.14 & 0.78 & 5000.00 & 0.04 & 5000.00 & 0.06 \\ 
\hline
\end{tabular}}
\caption{Results are obtained over 1000 replications with $T=10000$, $m=200$, $h=100$, $\tau=5000$, $d=100, c_{\textrm{Local}}=3.44$, and $\alpha=0.05$ for all three methods. The blue colours are labelled when both the false positive rates and average detection delay are small.}
\label{autocorrelation}
\end{table}
\end{landscape}
\section{Conclusion}
\label{conclusion}
Within this paper, we proposed an online communication-efficient distributed changepoint detection method, and it can achieve similar performance as an idealistic setting but save many transmission costs. Numerically, we show that the local threshold and window size have an impact on the performance of our algorithm, and there is a trade-off in choosing a local threshold and window size. In application, we recommend choosing a large local threshold in general cases. But when the change is extremely small, the choice of the local threshold depends on the communication and storage budgets. If the communication budget is much more limited, choosing a large threshold with a large window size is sensible. If the storage cost is much more expensive, choosing a small threshold with small window size will approximately achieve the idealistic performance. \\

The violation of independent assumptions will negatively affect the power of our proposed method. We tried to solve this problem by inflating thresholds or estimating the long-run variance. Both ways can, to some extent, improve our algorithm when the auto-correlation problem is not severe. However, both approaches fail to detect changes in highly auto-correlated data. Therefore, one of the future research directions is how to detect change within highly auto-correlated data in real-time.\\

\textbf{Acknowledgements}\\
Yang gratefully acknowledges the financial support of the EPSRC via the STOR-i Centre for Doctoral Training (EP/S022252/1). This research was also supported by EPSRC grant EP/R004935/1 (Eckley) together with financial support from BT Research (Eckley, Fearnhead, Yang). The authors are also grateful to Lawrence Bardwell who played a key role in inspiring this work, and Dave Yearling (BT) for several helpful conversations that helped shape this research.
\clearpage

\clearpage
\bibliography{bibliog}
\clearpage
\appendix
\appendix

\section{Proofs}

\subsection{Proof of Theorem \ref{thm:global_null}}
\label{sec:global_null}

Squaring and expanding the weighted global MOSUM statistic in Equation \eqref{eq:global_MOSUM} for the two different cases gives
\begin{align}
    \label{eq:global_MOSUM_expansion}      
    \left( w(k,h)\mathcal{T}(m,k,h) \right)^2 =
    \begin{cases}
      \sum_{i=1}^{d} \left( w(k,h)\mathcal{T}_{i}(m,k,h) \right)^2   &\quad\text{dense case,} \\
      \sum_{i=1}^{d} \left( w(k,h)\mathcal{T}_{i}(m,k,h) \mathbbm{1}_{ \{ w(k,h)\mathcal{T}_{i}(m,k,h) > c_{\textrm{Local}} \} } \right)^2  &\quad\text{sparse case.}
    \end{cases}
\end{align}
From Theorem \ref{thm:local} with $k = ht$, the weighted local MOSUM and its hard-thresholded counterpart have limit
\begin{align}
  \label{eq:limits_null}
  \begin{split}    
  \lim_{m \rightarrow \infty} \left( w(k,h)\mathcal{T}_i(m,k,h) \right)^2 &\overset{\mathcal{D}}{\longrightarrow} \rho(t)^2Z_i(t)^2, \\  
  \lim_{m \rightarrow \infty} \left( w(k,h)\mathcal{T}_i(m,k,h) \mathbbm{1}_{ \{ w(k,h)\mathcal{T}_{i}(m,k,h) > c_{\textrm{Local}} \} } \right)^2 &\overset{\mathcal{D}}{\longrightarrow} \rho(t)^2Z_i(t)^2\mathbbm{1}_{ \{  \rho(t)Z_{i}(t) > c_{\textrm{Local}} \} },
\end{split}
\end{align}                                                                                                                                                    
where $Z_i(t)$ is defined in Equation \ref{brownian}. Taking the limit of \eqref{eq:global_MOSUM_expansion} as $m \rightarrow \infty$ gives the result.

\subsection{Proof of Theorem \ref{thm:global_alt}}
\label{sec:global_alt}

It is enough to show that
\begin{align*}
  \left( w( \tilde{k}, h) \mathcal{T}( m,\tilde{k},h) \right)^2 \overset{\mathcal{P}}{\longrightarrow} \infty,
\end{align*}
for a time $\tilde{k}$ later than the change point $k^{*}$ but before the end of the monitoring time $\lfloor m\Tilde{T} \rfloor$.

Since $k^{*} \leq \lfloor h \nu \rfloor$, we can choose $\tilde{k} = \lfloor x_0h \rfloor + h$ where $\nu < x_0 < \Tilde{T}\frac{m}{h} - 1$ so that $k^{*} \leq \tilde{k} - h$.

If data stream $i$ is affected by the change, so that $i \in \mathcal{S}$ and $\delta_i \ne 0$ then
\begin{align*}
  \frac{1}{h}\mathcal{T}_{i}(m, \tilde{k}, h) &= \frac{1}{h \hat{\sigma}_i} \left|  \sum_{t = m + \lfloor x_0 h \rfloor + 1}^{ m + \lfloor x_0 h \rfloor + h } \left( X_{i,t} - \hat{\mu}_i \right)\right|  \\
                                              &= \frac{1}{h \hat{\sigma}_i} \left|  \sum_{t = m + \lfloor x_0 h \rfloor + 1}^{ m + \lfloor x_0 h \rfloor + h } \left( \mu_i + \delta_i + \epsilon_{i,t} - \hat{\mu}_i \right)\right|  \\
  &= \frac{1}{h \hat{\sigma}_i} \left| h(\mu_i - \hat{\mu}_i) + h\delta_i + \sum_{t = m + \lfloor x_0 h \rfloor + 1}^{ m + \lfloor x_0 h \rfloor + h }\epsilon_{i,t}  \right|  \\
                                              &= \frac{1}{\hat{\sigma}_i} \left| \mu_i - \hat{\mu}_i + \delta_i + \frac{1}{h}\sum_{t = m + \lfloor x_0 h \rfloor + 1}^{ m + \lfloor x_0 h \rfloor + h }\epsilon_{i,t} \right|   \\
  &= \frac{1}{\hat{\sigma}_i} \left| \delta_i + o_{P}(1) \right|. 
\end{align*}
On the other hand if $i \notin \mathcal{S}$,
\begin{align*}
  \frac{1}{h}\mathcal{T}_{i}(m, \tilde{k}, h) &= \frac{1}{\hat{\sigma}_i} \left| o_{P}(1)  \right|.
\end{align*}

For the global dense procedure
\begin{align*}
  \left( w( \tilde{k}, h) \mathcal{T}( m,\tilde{k},h) \right)^2 &= w(\tilde{k}, h)^2 \sum_{i=1}^{d} \mathcal{T}_i(m, \tilde{k}, h)^2 \\
  &= \left( \frac{1}{\sqrt{h}} \rho \left( \frac{\tilde{k}}{h} \right) \right)^2 \times h^2 \times \sum_{i=1}^{d} \left( \frac{1}{h} \mathcal{T}_i(m, \tilde{k}, h) \right)^2 \\
                                                                &= h \rho \left( \frac{\tilde{k}}{h} \right)^2 \sum_{i=1}^{d} \left( \frac{1}{h} \mathcal{T}_i(m, \tilde{k}, h) \right)^2 \\
  &= h \rho \left( x_0 + 1 + o(1) \right)^2 \sum_{i \in \mathcal{S}} \left( \frac{\delta_i}{\hat{\sigma}_i} \right)^2 + o_{P}(1) \overset{\mathcal{P}}{\longrightarrow} \infty,
\end{align*}
as $m,h \rightarrow \infty$.

For the global sparse procedure the local MOSUM's $w( \tilde{k}, h) \mathcal{T}( m,\tilde{k},h)$ are hard thresholded. Since these diverge to infinity individually then the same argument used for the dense procedure applies.

\end{document}